\documentclass[11pt,letterpaper]{article}
\usepackage[margin=1in]{geometry}
\usepackage{amsmath,amssymb,amsthm,mathtools}
\usepackage{booktabs,array,float}
\usepackage{xcolor}
\usepackage{enumitem,needspace,microtype}
\usepackage[backend=biber,style=alphabetic,sorting=nyt,sortcites=true,
  maxbibnames=99,maxalphanames=4,giveninits=true,doi=true,url=false,
  isbn=false,eprint=true]{biblatex}
\AtEveryBibitem{\ifentrytype{online}{\clearfield{doi}}{}}
\usepackage[hidelinks]{hyperref}
\usepackage[nameinlink,noabbrev,capitalise]{cleveref}
\allowdisplaybreaks
\usepackage{aliascnt}
\newtheorem{theorem}{Theorem}[section]
\newaliascnt{lemma}{theorem}
\newtheorem{lemma}[lemma]{Lemma}
\aliascntresetthe{lemma}
\newaliascnt{corollary}{theorem}
\newtheorem{corollary}[corollary]{Corollary}
\aliascntresetthe{corollary}
\newaliascnt{proposition}{theorem}
\newtheorem{proposition}[proposition]{Proposition}
\aliascntresetthe{proposition}
\theoremstyle{definition}
\newaliascnt{definition}{theorem}

\aliascntresetthe{definition}
\theoremstyle{remark}
\newaliascnt{remark}{theorem}

\aliascntresetthe{remark}
\crefname{theorem}{Theorem}{Theorems}
\crefname{lemma}{Lemma}{Lemmas}
\crefname{corollary}{Corollary}{Corollaries}
\crefname{proposition}{Proposition}{Propositions}
\crefname{definition}{Definition}{Definitions}
\crefname{remark}{Remark}{Remarks}
\newcommand{\FB}{\operatorname{FB}}
\newcommand{\SB}{\operatorname{SB}}
\newcommand{\GFT}{\operatorname{GFT}}
\newcommand{\cF}{\mathcal F}
\newcommand{\cX}{\mathcal X}

\newcommand{\1}{\mathbf 1}
\newcommand{\E}{\mathbb E}
\newcommand{\dd}{\,\mathrm d}

\newcommand{\pos}[1]{(#1)_+}
\hypersetup{pdftitle={From Bilateral Trade to Matching Markets: Sharp Gains from Trade},
  pdfauthor={Zhengyang Liu, Ying Qin, Zihe Wang}}
\title{From Bilateral Trade to Matching Markets:\\ Sharp Gains from Trade}
\author{%
Zhengyang Liu\thanks{Beijing Institute of Technology. Email: \texttt{zhengyang@bit.edu.cn}}
\and
Ying Qin\thanks{Renmin University of China. Email: \texttt{qinying0420@ruc.edu.cn}}
\and
Zihe Wang\thanks{Renmin University of China. Email: \texttt{wang.zihe@ruc.edu.cn}}
}
\date{}

\begin{document}
\pagenumbering{roman}
\maketitle

\begin{abstract}
\normalsize
We study gains from trade in matching markets with independent private
values and costs, Bayesian incentive compatibility, interim individual
rationality, and no expected budget deficit.   A  second-best
  guarantee for finite bilateral trade extends  without loss to matching
markets with  independent Borel priors,  arbitrary downward-closed
feasibility, and finite expected first-best gains. For bounded buyers with monotone hazard rates and
arbitrary bounded sellers, we determine the exact worst-case ratio of
second-best to first-best gains, approximately $0.72490721$. For binary
buyers and sellers with at most $m$ types, we determine the exact ratio for
every $m$, including $8/9$ when $m=2$ and a limit of $4/5$ as $m$ grows.
Both families of bounds are tight already in bilateral trade.
 \end{abstract}
\clearpage
\setcounter{tocdepth}{2}
\tableofcontents
\clearpage
\pagenumbering{arabic}
\section{Introduction}
A buyer values an item at $v$, and its seller incurs cost $c$ by giving it up.
Trading creates \emph{gains from trade} (GFT) of $v-c$; keeping the item creates
zero gains. This one-buyer, one-seller problem is \emph{bilateral trade}.
In a \emph{matching market}, a bipartite graph specifies which buyers can
trade with which sellers. Each agent participates in at most one trade, and
the gains of a matching are the sum of its edge gains. If values and costs
were known, a maximum-weight feasible matching would maximize these gains.
Its expectation is the \emph{first-best} benchmark, denoted by $\FB$.

Private information changes the problem: a mechanism must elicit values and
costs, choose trades, and set payments. Each agent's private value or cost,
called her \emph{type}, is drawn independently from a known distribution.
We require Bayesian incentive compatibility (BIC): truthful reporting
maximizes each type's expected utility when the other agents report
truthfully. Interim individual rationality (IIR) requires that expected
utility to be nonnegative. Ex-ante weak budget balance (WBB) requires
nonnegative expected net revenue, so the mechanism needs no subsidy on
average. The \emph{second-best} benchmark $\SB$ is the supremal expected
GFT subject to these three conditions. Payments may have either sign, and
there is one expected-budget constraint for the whole market, not one per
trade. These constraints can prevent first best even in bilateral
trade~\cite{MS83}.

In our separate work~\cite{LQRW26}, we establish that the worst-case ratio $\SB/\FB$ in bilateral trade is exactly $1/2$.
Bei, Li, Wu, and Zhou~\cite{BLWZ26} establish the same tight ratio in
matching markets. This equality raises a more general question:
\begin{quote}
  When local restrictions on the type distributions improve bilateral efficiency, do competing trades create an additional loss?
\end{quote}
We give a bilateral-to-matching reduction that preserves
the approximation factor under precisely stated changes of the local
distributions. Two applications determine sharp factors from a restriction on the shape
of buyer distributions and from a bound on the number of types. In each case, the worst ratio
in a matching market is already approached by a single buyer and seller.

\subsection{Main Results}
Our reduction separates a bilateral efficiency bound from the problem of
combining trades. Two applications then quantify the improvement from
monotone hazards and from small supports. They share the composition step,
but require different bilateral bounds.

For a mechanism $\mathcal M$, let $X$ be its random matching under truthful
reports and let $p_a$ be agent $a$'s net payment to the mechanism. Define
\[
 \GFT(\mathcal M)=\E\Bigl[\sum_{(i,j)\in X}(v_i-c_j)\Bigr],
 \qquad \operatorname{Rev}(\mathcal M)=\E\Bigl[\sum_a p_a\Bigr].
\]
Both expectations include the types and the mechanism's randomization.
To handle the shared budget, fix a multiplier $\alpha\ge0$ and consider
\[
 V_\alpha=\sup_{\mathcal M\ {\rm BIC,IIR}}
       \{\GFT(\mathcal M)+\alpha\operatorname{Rev}(\mathcal M)\},
\]
without imposing budget balance inside the supremum. For finite types,
linear programming duality gives $\SB=\min_{\alpha\ge0}V_\alpha$.
Thus bounds for all multipliers yield a second-best guarantee. The
parameter $\alpha$ measures how much weight we place on net revenue; the
function $\alpha\mapsto V_\alpha$ is the \emph{Lagrangian curve}.
A seller's \emph{virtual cost} $\tau(s)$ accounts for actual cost and
\emph{information rents},   {the utilities needed to make truthful
reporting optimal}; the finite formula
is given in \cref{eq:finite-virtual-scores}. Following \cite{LQRW26}, we call
a seller \emph{$\alpha$-weakly regular} if the combined cost
$s+\alpha\tau(s)$ is nondecreasing in $s$. This condition lets us order
sellers' costs and incentive payments consistently. Feasibility is
\emph{downward-closed} when deleting trades preserves feasibility.
The reduction below transfers a bilateral bound at each fixed multiplier.

\Needspace{6\baselineskip}
\begin{theorem}[Bilateral-to-Matching Reduction]\label{thm:bilateral-to-matching}

 Fix $\alpha,\beta\ge0$. If $V_\alpha\ge\beta\FB$ for every independent finite
bilateral-trade instance, then the same bound  holds for every independent finite
matching market with downward-closed feasibility.

 \end{theorem}

The premise ranges over every buyer distribution. For restricted results,
we vary probabilities on a fixed buyer \emph{support}, the set of possible
types. Following \cite[Definition~4.8]{BRTW26},   {a \emph{cap} is  an
initial set of seller costs: for example, the cap through $s_2$   contains
 $s_1,s_2$ from $s_1<s_2<s_3$. }  Its boundary may be included or excluded.
Bilateral bounds use the conditional seller prior below the cap, with
probabilities renormalized; composition uses the original probability weights.

A family of local rules is \emph{cap-monotone} if, at fixed buyer and seller
reports, enlarging the cap never decreases the trade probability. This
extends the cited definition to randomized rules. Our construction derives
such a family from bilateral value bounds (\cref{thm:compatible-targets}),
after regularizing each seller's full distribution once
(\cref{prop:regularization-main}).

  \begin{corollary}[Second-Best Bounds]
 \label{thm:finite-main-v2}\label{cor:reduction-consequences}

 For $\gamma\ge0$,   if every independent finite bilateral-trade instance satisfies
 $\SB\ge\gamma\FB$, then  every independent finite matching market with
 downward-closed feasibility  satisfies the same bound. The bound also holds
for independent Borel type distributions with
$\FB<\infty$. In particular, $\SB\ge\FB/2$, and the factor is tight
already bilaterally.

 \end{corollary}
\begin{proof}

For finite priors, $V_\alpha\ge\SB\ge\gamma\FB$ in every bilateral
market. Apply  \cref{thm:bilateral-to-matching} for each $\alpha$ and use
$\SB=\inf_{\alpha\ge0}V_\alpha$. The  extension from finite matching priors
to Borel priors with finite expected first best is proved in
\cref{lem:finite-borel-efficiency-transfer}. The  factor $1/2$ and its
bilateral tightness follow from our separate work~\cite{LQRW26}, recovering
the matching-market half guarantee of \cite{BLWZ26}.

 \end{proof}

  \paragraph{Application: generalized offering mechanisms.}
The uniform mixture of the generalized seller- and buyer-offering mechanisms
has a particularly simple comparison at $\alpha=1$. For every allocation $x$,
\[
 \GFT(x)+\Lambda(x)
 =\E\!\left[\sum_{(i,j)}(\phi_i(v_i)-c_j)x_{ij}\right]
 +\E\!\left[\sum_{(i,j)}(v_i-\tau_j(c_j))x_{ij}\right].
\]
The two terms are exactly the objectives underlying the generalized seller-
and buyer-offering rules, so \cref{lem:offering-comparison} gives
$2G_{\rm offer}\ge V_1$. Consequently, any transferred bound
$V_1\ge\beta\FB$ yields $G_{\rm offer}\ge\beta\FB/2$. The value
$\alpha=1$ is special because only then does
$(v-c)+\alpha(\phi-\tau)$ split into these two unweighted offering
objectives; for a general multiplier there is no analogous decomposition
for the same pair of mechanisms.

  {The  extension of the half guarantee  to independent Borel priors with finite
expected first-best gains  is \cref{thm:main-v2}. }
 The next two results ask how local type restrictions improve this factor. We first study \emph{monotone hazard rates} (MHR), a distributional condition
used in earlier bilateral-trade bounds~\cite{BM16,Fei22}.   For a buyer with
 density $f$ and distribution function $F$,   MHR means that the hazard rate
 $f(x)/(1-F(x))$ is nondecreasing   on the interior of the support. An  atom at
the largest value is allowed. This application constructs rules directly
on continuous report intervals and then uses the same composition argument.

\begin{theorem}[MHR Buyers]\label{thm:mhr-intro}
For independent bounded priors, MHR buyers and arbitrary sellers satisfy
$\SB\ge\gamma_{\rm MHR}\FB$ under every downward-closed matching constraint,
where $\gamma_{\rm MHR}\simeq0.72490721$ is defined exactly in
\cref{mhr:eq:gamma}. The factor is sharp already in bilateral trade.
\end{theorem}

This determines the exact factor between the previous bilateral lower and
upper bounds $1/(e-1)$ and $2/e$~\cite{Fei22,BM16}, and shows that competing
trades cause no further loss. \Cref{thm:sharp-mhr} gives the sharp factor for
every $V_\alpha$. The same worst-case factors hold if buyers must have densities and hence
no atoms. We also extend the lower bounds to unbounded types when all types
have a common finite lower bound and expected first-best gains are finite
(\cref{app:clipping}). The formal statement in \cref{thm:sharp-mhr} specifies
the continuous report model.

Our second application asks how the number of possible types limits the
efficiency loss, without a hazard or other shape restriction. A \emph{binary} agent
  {has two support points; the support points  and probabilities may differ across
agents. }  Varying probabilities and conditioning on caps preserve support-size
bounds, so this application uses the finite reduction directly.

\Needspace{9\baselineskip}
\begin{theorem}[Binary Buyers]\label{thm:binary-intro}
For each integer $m\ge1$, the sharp second-best factor for independent finite
priors with binary buyers and at most $m$ types per seller, under
downward-closed matching constraints, is
\[
 \gamma_m=\left[\max_{0\le q<1}(1+q-q^2-q^m)\right]^{-1}.
\]
In particular, $\gamma_1=1$, $\gamma_2=8/9$, $\gamma_3=27/32$, and
$\gamma_m\downarrow4/5$. Each bound is tight bilaterally.
\end{theorem}

\Cref{thm:curve} gives the full multiplier curve and rational tight instances
in $[0,1]$. Exchanging the sides preserves the result, so one binary side
suffices for the tight $4/5$ guarantee with arbitrary finite support on the
other side. \Cref{sec:further-consequences} discusses refinements   and
{ guarantees for specific offering mechanisms }. All ratios exclude
$\FB=0$. The two applications identify different bilateral improvements;
the reduction explains why matching competition preserves them.

\subsection{Technical Overview}
The main obstacle is incentive compatibility when a buyer's report changes
both her chance to trade and the local distribution she faces. Fix an edge
and the reports of all agents outside that edge. First best selects the edge
for an initial set of seller costs, ending at a threshold determined by the
buyer's report. This set expands as the buyer's value rises. Applying a bilateral guarantee separately
at each cap need not work: the chosen allocations may fail to be monotone
across caps, and changing the allocation at lower buyer reports changes the
information rent owed to higher types.

For finite types, the minimax theorem turns a bound valid for every buyer
distribution on a fixed support into a single allocation satisfying an
inequality for each buyer type. We call such inequalities \emph{typewise}
bounds, to distinguish them from bounds averaged over the buyer's prior.
We meet them recursively, allocating the smallest mass of the cheapest
seller types that covers the current target and the rents already required
at lower buyer types. If a larger cap selected less mass, that mass would
also meet the smaller cap's requirements, contradicting minimality. The
resulting rules are cap-monotone. On first-best edges, lower buyer reports
use smaller caps and hence require no more rent than the fixed-cap bound
allows. Summing the bounds therefore loses no factor. Irregular sellers are
regularized once on their full distributions, preserving $V_\alpha$ and
weakly increasing $\FB$; regularizing each cap separately would not suffice.

The two applications supply different local bounds. For MHR buyers,
\emph{reflection} exchanges the two sides: for a common upper bound $U$,
a buyer value $b$ becomes a seller cost $U-b$, and a seller cost $c$ becomes
a buyer value $U-c$. This preserves gains and turns the MHR buyer prior
into a seller prior. A cutoff defined by a weighted-average identity
balances the allocation's gains and information rents. A comparison with an exponential distribution, in which the two cumulative
distribution functions cross at most once, proves the required typewise
inequality and cap-monotonicity directly, without discretization. For binary
buyers, a seller-cutoff recurrence gives the factor $1-q^m$, where
$q=\alpha/(1+\alpha)$. Accounting for the buyer information rent adds
$q(1-q)$, yielding the denominator in
\cref{thm:binary-intro}. Explicit bilateral families match both bounds.
Thus the finite reduction and the MHR analysis use the same composition
argument, but different proofs of the local inequalities.

\subsection{Related Work}
Table~\ref{tab:comparison} separates the known tight half guarantees from
the sharper distribution-dependent factors and the reduction proved here.
All efficiency ratios in the table compare second best with first best.
\begin{table}[htbp]
\centering
\begin{tabular}{@{}>{\raggedright\arraybackslash}p{.22\linewidth}>{\raggedright\arraybackslash}p{.34\linewidth}>{\raggedright\arraybackslash}p{.38\linewidth}@{}}
\toprule
Setting & Previous Results & This Paper \\
\midrule
Arbitrary independent types
 & Tight $1/2$ bilaterally (our separate work~\cite{LQRW26}); tight $1/2$ in matching markets~\cite{BLWZ26}
 & Recovers the matching bound through a general bilateral-to-matching reduction \\
\addlinespace
MHR buyer; arbitrary seller
 & Bilaterally, $1/(e-1)\le\inf\SB/\FB\le2/e$~\cite{Fei22,BM16}
 & Exact factor $\gamma_{\rm MHR}\simeq0.7249$, both bilaterally and in matching markets \\
\addlinespace
Binary buyer and seller
 & Bilateral numerical grid minimum $0.89189$~\cite{Schott23}
 & Exact factor $8/9$ in both settings; sharp formula for every seller support size \\
\addlinespace
Bilateral-to-matching reduction
 & Requires supplied cap-monotone rules~\cite{BRTW26}
 & Derives such rules from bilateral Lagrangian inequalities \\
\bottomrule
\end{tabular}
\caption{Independent types and the BIC/IIR benchmark with one expected
budget. The half factors are prior results. Fei's seller-pricing guarantee
implies the stated MHR lower bound; the reported binary grid minimum is
numerical, rather than an exact worst-case factor.}
\label{tab:comparison}
\end{table}

Our separate bilateral half theorem and local allocation construction
\cite{LQRW26} provide the starting point. Bei, Li, Wu, and Zhou
\cite{BLWZ26} prove the tight matching-market half theorem. Babaioff,
Rubinstein, Tan, and Wang give a \emph{meta-auction}, a matching-market
mechanism that combines supplied cap-monotone bilateral rules~  {\cite{BRTW26}}. Our local theorem derives the
needed rules from bilateral value inequalities and preserves gains plus
revenue at every multiplier. This stronger reduction yields the sharp MHR
and finite-support factors, beyond the already known half guarantee.

For MHR buyers, Fei's tight $1/(e-1)$ seller-pricing guarantee is not tight
for second best: its zero-cost-seller example has $\SB=\FB$~\cite{Fei22}.
Blumrosen and Mizrahi give the $2/e$ second-best upper bound, even with
both priors MHR; their $2/e$ seller-offering lower bound additionally
requires a concave buyer hazard~\cite[Theorems~3.10 and 4.1, full version]{BM16}.
We impose MHR on buyers alone. Schottm\"uller numerically evaluates the
second-best benchmark for fixed binary priors within an information-design
study~\cite{Schott23}. We determine its exact infimum and recover the
reported grid instance in \cref{app:tight-family}.

Random offerer chooses the buyer or seller uniformly to make a
profit-maximizing offer. Its guarantees~\cite{DMSW22,Fei22,HW25} and
upper-bound examples~\cite{BDK21,CGLM26} were improved by Jo~\cite{Jo26}
to $1/\pi\le\inf\operatorname{RO}/\FB<0.460243$. The connection to our
reduction is a separate $V_1$ comparison with generalized offering mechanisms
(\cref{app:offering-consequences}); random-offerer upper bounds do not bound
$\SB$. Robust bilateral mechanisms and stronger incentive requirements
appear in \cite{HR87,BD21}. Fixed-price welfare bounds~\cite{KPV22,CW23,LRW23,GK26}
include the seller's endowment, so they do not imply the same multiplicative
GFT bounds. Sample-based offers~\cite{DMSWW25} use a different access model
from our known priors.

For larger markets, trade reduction~\cite{McA92} and modular double
auctions~\cite{DTR17} give structured mechanisms; strongly budget-balanced
welfare approximations appear in \cite{CKLT16,CGKLRT20}.
Generalized offering mechanisms approximate second-best GFT~\cite{BCWZ17},
and \cite{BCGZ18} also obtains asymptotic first-best efficiency. We instead
seek exact worst-case factors without a large-market limit.
Multi-dimensional GFT~\cite{CGMZ21,RTZ26} permits private item-specific
values; our types remain scalar, even with public match-specific values.
Myerson's payment identities~\cite{Mye81} and the budget conversions
in \cite{BCWZ17,BN09} underlie our formulation.

\paragraph{Organization.}
\Cref{sec:model} sets out the model and budget dual.
\Cref{sec:reduction} proves the reduction in three steps: composition,
construction of local rules, and seller regularization.
\Cref{sec:mhr-overview,hs:sec:sharp} then establish the sharp MHR and
finite-support bounds. \Cref{sec:further-consequences} discusses consequences
and open questions. \Cref{app:foundations,app:general-priors} supply the
deferred finite proofs and the extensions to general priors.
\Cref{app:refinements,app:further-scope} develop the additional refinements
and applications discussed in \cref{sec:further-consequences}.

\section{  {Preliminaries} }\label{sec:model}
To prove the reduction, we need an allocation-based expression for the
shared budget constraint. We first give the finite model and its payment
identities, then explain the continuous report convention used for MHR buyers.
A market $\mathcal I$ has a bipartite compatibility graph $G=(B,S,E)$. A buyer receives
at most one item and values every compatible item equally; a seller supplies
one item. Let $\mathcal F$ be the allowed matchings. Throughout the structural results, it is any nonempty downward-closed
family of matchings of $G$. Priors and feasibility are public and do not depend on reports.

Buyer $i$ has support $b_i^1<\cdots<b_i^{n_i}$   with probabilities  $g_{i\ell}>0$;
seller $j$ has support $s_j^1<\cdots<s_j^{m_j}$   with probabilities  $f_{jk}>0$.
An allocation $x$ assigns at every profile a   probability distribution  over $\mathcal F$, with edge probabilities

 $x_{ij}$. A report profile $\theta$ lists all agents' reports, and
$\theta_{-i}$ omits agent $i$. Write $q_i=\sum_jx_{ij}$, $y_j=\sum_ix_{ij}$ and
$Q_i(b)=\E_{\theta_{-i}}[q_i(b,\theta_{-i})]$,
$Y_j(s)=\E_{\theta_{-j}}[y_j(s,\theta_{-j})]$.
Here $q_i,y_j$ are trade probabilities at a report profile, while $Q_i,Y_j$
average over the other agents' types and are called \emph{interim} trade
probabilities. Agents   maximize  expected utility. Utilities are
\emph{quasi-linear}: a signed payment $p_a$ is from agent $a$ to the mechanism,
and buyer and seller utilities are $v_iq_i-p_i$ and $-c_jy_j-p_j$,
respectively.
BIC and IIR have the meanings given in the introduction. Universal
dominant-strategy incentive compatibility (universal DSIC) requires truthful reporting for every fixed report-independent random seed
and every opponents' report profile; ex-post IR requires nonnegative
truthful utility at every profile and random realization.

We use $(w)_+=\max\{w,0\}$ for the positive part of  $w$.
Since GFT depends only on the allocation, we also write $\GFT(x)$ for the
expected gains of any mechanism with allocation rule $x$. Our benchmarks are
\begin{equation}\label{eq:gft-and-first-best}
 \GFT(x)=\E\Bigl[\sum_{(i,j)\in E}(v_i-c_j)x_{ij}\Bigr],\qquad
 \FB=\E\Bigl[\max_{M\in\mathcal F}\sum_{(i,j)\in M}(v_i-c_j)\Bigr].
\end{equation}
\begin{equation}\label{eq:second-best-definition}
 \SB=\max\{\GFT(x):(x,p)\text{ is feasible, BIC, IIR, ex-ante WBB}\}.
\end{equation}
Ex-ante WBB requires $\E[\sum_a p_a]\ge0$. Strong budget balance (SBB) below
means $\sum_ap_a=0$ at every report profile and internal realization. Signed
transfers impose no limited-liability restriction. The finite-support linear program
ensures attainment of $\SB$.

Let $\mathcal X$ be the polytope of feasible allocations with $Q_i$ nondecreasing
and $Y_j$ nonincreasing. These monotonicity conditions are exactly those
needed for BIC implementability with suitable payments. The \emph{virtual
value} $\phi_i$ and \emph{virtual cost} $\tau_j$ express the corresponding
expected payments as weighted trade probabilities:
\begin{equation}\label{eq:finite-virtual-scores}
 \begin{aligned}
 \phi_i(b_i^\ell)&=b_i^\ell-(b_i^{\ell+1}-b_i^\ell)
          \frac{\sum_{r>\ell}g_{ir}}{g_{i\ell}}, & \phi_i(b_i^{n_i})&=b_i^{n_i},\\
 \tau_j(s_j^k)&=s_j^k+(s_j^k-s_j^{k-1})
          \frac{\sum_{r<k}f_{jr}}{f_{jk}}, & \tau_j(s_j^1)&=s_j^1.
 \end{aligned}
\end{equation}
The buyer formula applies below her largest type and the seller formula
above her smallest type. Each correction accounts for the utility that
other types   {must receive for truthful reporting}. The strict tails account for
the discrete gaps between consecutive reports. The next lemma uses these
payment identities to replace budget feasibility by one linear inequality.
\begin{lemma}[Payments and Budget Duality]\label{lem:allocation-characterization}
For $x\in\mathcal X$, the largest expected net revenue compatible with BIC and
IIR is
\begin{equation}\label{eq:finite-revenue-identity}
 \Lambda(x)=\E\Bigl[\sum_i\phi_i(v_i)q_i\Bigr]-\E\Bigl[\sum_j\tau_j(c_j)y_j\Bigr].
\end{equation}
An allocation is implementable with BIC, IIR, and ex-ante WBB exactly when it
belongs to $\mathcal X$ and $\Lambda(x)\ge0$. Moreover,
\begin{equation}\label{eq:lagrange-duality}
 \SB=\min_{\alpha\ge0}V_\alpha,\qquad
 V_\alpha=\max_{x\in\mathcal X}\{\GFT(x)+\alpha\Lambda(x)\}.
\end{equation}
The minimum is attained. Some optimum is a mixture of at most two Lagrangian
maximizers at a common minimizing multiplier.
\end{lemma}
The lemma reduces the second-best problem to optimizing a weighted sum
of gains and revenue and then choosing the multiplier. Write
$\mathcal L_\alpha(x)=\GFT(x)+\alpha\Lambda(x)$ for the
Lagrangian of an allocation; when a local prior is specified, the same
notation uses that prior in both expectations.

The finite-type payment identities and linear-programming argument are proved in
\cref{sec:finite-calculus}.

For a fixed multiplier, set $a_i^\alpha(b)=b+\alpha\phi_i(b)$ and
$d_j^\alpha(s)=s+\alpha\tau_j(s)$. These scores combine actual gains with
incentive payments. They need not increase with the type. To restore
monotonicity, \emph{ironing}   repeatedly merges adjacent blocks whose
 probability-weighted means   violate monotonicity, producing nondecreasing
scores $\bar a_i^\alpha,\bar d_j^\alpha$.
A seller is $\alpha$-weakly regular exactly when $d_j^\alpha$ is already
nondecreasing, so its ironing changes nothing. For each profile let $X^\alpha$ maximize
$\sum_{(i,j)\in M}(\bar a_i^\alpha(v_i)-\bar d_j^\alpha(c_j))$ over
$\mathcal F$, then minimize cardinality and use a fixed lexicographic order
on edge sets. The next lemma shows that this matching rule solves the
Lagrangian problem while   {preserving}  the monotonicity needed for incentives.
\begin{lemma}[Fixed-Multiplier Optimization]\label{lem:exact-ironed-optimizer}
The rule $X^\alpha$ is deterministic and attains $V_\alpha$. With all
other reports fixed, each buyer's trade probability is nondecreasing in her
value and each seller's is nonincreasing in her cost (\emph{pointwise
monotonicity}).
It trades only when $v_i>c_j$.  { For rational inputs and $\alpha$, its  evaluation uses one call to an oracle that maximizes edge weight over $\mathcal F$. }  With opponents fixed, an agent's partner is fixed throughout the
reports at which she trades.
\end{lemma}
\noindent Ironing upper-bounds every interim-monotone allocation; $X^\alpha$
is constant on every ironing block and attains equality.   {Ironing also gives}
 $\bar a_i^\alpha(b)\le(1+\alpha)b$ and
$\bar d_j^\alpha(s)\ge(1+\alpha)s$, so positive edge weights imply positive
actual surplus. The fixed tie rule makes the best served matching independent
of the varying own coefficient. Details are in \cref{app:optimizer}.

For independent types and unrestricted signed transfers, ex-ante WBB and
profilewise SBB have the same allocation benchmark~\cite{BN09,BCWZ17}.
\Cref{app:sbb-conversion} recalls the conversion; it does not preserve
universal DSIC or ex-post IR.

\paragraph{Continuous Report Intervals.}
  For the continuous MHR analysis we use the analogous bounded-interval
formulation, with monotone  interim trade probabilities  implemented by the
  envelope formulas in \cref{eq:borel-envelope}; general Borel priors are treated
in \cref{app:general-priors}.

\section{The Bilateral-to-Matching Reduction}\label{sec:reduction}
The budget dual reduces our task to preserving a bilateral bound for each
multiplier. We first show how to combine suitable local allocation rules,
then construct these rules from finite bilateral value bounds, and finally
remove the regularity assumption on sellers. The MHR application will use
the same composition step with a direct continuous construction.

\subsection{Composing Local Allocation Rules}\label{sec:cap-composition}
Fix $\alpha\ge0$. Start with a \emph{canonical} first-best matching:
among maximum-GFT matchings, choose one of minimum cardinality and then
use a fixed, report-independent order to break remaining ties. The next
lemma explains why its edges define consistent local problems.
\begin{lemma}[First-Best Thresholds]\label{lem:stable-partner-v2}
Use a first-best matching with minimum cardinality and a fixed final order.
With all other reports fixed, its matching is constant throughout the
reports at which an agent trades (her \emph{winning reports}). For an edge, fix the reports of all agents outside it. The accepted seller
types form an initial set, increasing with the buyer report. Every selected edge
has positive surplus.
\end{lemma}
\begin{proof}
Every matching serving a buyer has weight equal to her report plus a fixed
intercept. The best served matching is therefore fixed; it crosses the best
unserved constant at most once. Sellers have slope $-1$ instead. Lowering a
selected seller cost or raising the selected buyer value preserves the same
matching, proving the stated monotonicity and nesting properties. Deleting a nonpositive-surplus edge
improves the objective or reduces cardinality and is feasible by downward closure.
\end{proof}

The lemma lets us replace each first-best edge by a local trade decision
without changing an agent's candidate partner as she changes a winning
report. Fix an edge and the reports of all agents outside it; denote this
\emph{residual profile} by $\omega$. We suppress the edge and $\omega$
until the market-level sum. A cap $a$ is an initial
set of seller reports, possibly empty. It consists of the costs below a
threshold, with the boundary either included or excluded. Write
$z^a(b,s)\in[0,1]$ for a local trade probability, zero outside $a$, and
$m^a(b)=\E_s[z^a(b,s)]$. This expectation uses the original seller
probabilities: we do not divide by $\Pr[s\in a]$.

We next express the local objective using incentive-compatible payments.
The \emph{buyer envelope} is the minimum utility required by truthfulness
when the lowest buyer type receives zero utility. For a trade-probability
function $m$, write this utility as
\[
 (\mathcal Bm)(b_\ell)=\sum_{t<\ell}(b_{t+1}-b_t)m(b_t)
 \quad\hbox{or}\quad
 (\mathcal Bm)(b)=\int_L^b m(u)\,\dd u,
\]
for a finite support or a full bounded report interval $[L,U]$, respectively.
For fixed $b$, if $z(b,\cdot)$ is nonincreasing, the corresponding
seller receipt (money paid to the seller) is
\[
 P_S[z](b,s_k)=s_kz(b,s_k)+\sum_{h>k}(s_h-s_{h-1})z(b,s_h),
 \quad\hbox{or}\quad
 P_S[z](b,s)=sz(b,s)+\int_s^U z(b,u)\,\dd u.
\]
The seller formula sets the highest seller type's utility to zero.
These payment formulas, which we call \emph{normalized envelope payments},
use the full report domain, including reports where $z$ is zero. At a fixed
buyer report $b$, the expected buyer payment is $bm^a(b)-(\mathcal Bm^a)(b)$.
Thus gains plus $\alpha$ times net revenue equal
\begin{equation}\label{eq:fixed-cap-lagrangian}
 C^a(b)=(1+\alpha)b\,m^a(b)
       -\E_s\bigl[sz^a(b,s)+\alpha P_S[z^a](b,s)\bigr]
       -\alpha(\mathcal Bm^a)(b).
\end{equation}

We will lower-bound this expression separately for each buyer report.
The following theorem shows that these local bounds add across the
first-best matching, provided the rules are cap-monotone.

\begin{theorem}[Edge-by-Edge Composition]\label{thm:edge-composition}
Consider independent scalar priors, either on finite report supports or on
full common bounded report intervals, with downward-closed matching
feasibility. For each edge and residual profile, suppose a family $z^a$
is nondecreasing in $b$, nonincreasing in $s$, and coordinatewise
nondecreasing as the cap expands. Suppose also $C^a(b)\ge T^a(b)$ for
specified targets, with the empty-cap allocation and target zero.
In the interval model require joint Borel measurability (so the expectations are defined) in residual
reports, the threshold and whether it is included, and the buyer and seller
reports.
If $a_{e,\omega}(b)$ is the first-best cap, then
\begin{equation}\label{eq:composition-bound}
 V_\alpha(\mathcal I)\ge
 \sum_{e=(i,j)}\E_{\omega,b_i}\bigl[T^{a_{e,\omega}(b_i)}_{e,\omega}(b_i)\bigr].
\end{equation}
In particular, targets $T^a(b)=\beta\E_s[(b-s)_+\1\{s\in a\}]$
give $V_\alpha\ge\beta\FB$ with no loss in $\beta$.
\end{theorem}
\begin{proof}
Set $x_e(b,s,\omega)=z^{a_{e,\omega}(b)}(b,s)$ and   {keep}  each edge of the canonical
first-best matching with this probability. This is feasible by
downward closure. \Cref{lem:stable-partner-v2} and the two monotonicities
give monotone trade probabilities in every own report, hence the normalized envelope payments
implement BIC and IIR. Fix the edge and residual profile and write
$\widetilde m(w)=m^{a(w)}(w)$. At current buyer report $b$, the seller
receipt and cost terms equal those for the fixed cap $a(b)$. The actual
conditional Lagrangian is therefore
\begin{equation}\label{eq:unified-rent-correction}
 C^{a(b)}(b)+\alpha\mathcal B
       \bigl[m^{a(b)}(\cdot)-\widetilde m(\cdot)\bigr](b)
 \ \ge\ C^{a(b)}(b).
\end{equation}
Indeed, $a(w)\subseteq a(b)$ for $w<b$, and $\mathcal B$ has nonnegative
weights. Envelope payments add over edges. Independence leaves the endpoint
priors unchanged after residual reports are fixed; averaging and summing
prove \cref{eq:composition-bound}. Boundedness justifies all integrals.
For the stated surplus targets, each selected edge has positive surplus,
so their sum is exactly $\beta\FB$.
\end{proof}
This theorem composes local allocation rules satisfying the stated typewise
Lagrangian bounds; it does not compose arbitrary bilateral mechanisms. The next
step shows when finite value inequalities alone force such rules;
\cref{sec:mhr-overview} instead constructs them analytically. Passing from
fixed multipliers to an expected-budget guarantee uses finite duality or,
for bounded intervals, the separation argument in \cref{mhr:sec:cap-composition}.

\subsection{Constructing the Local Rules}\label{sec:cap-monotonicity}\label{app:finite-compatibility}\label{app:automatic-proof}
The composition theorem needs one family of rules whose probabilities
increase as the seller cap expands. We now show that, for finite
types, bilateral value bounds alone guarantee such a family.
Fix $\alpha$, a buyer support $b_1<\cdots<b_n$, and a seller with costs
$s_1<\cdots<s_m$ and masses $f_k>0$. Put $F_0=0$, $F_r=\sum_{k\le r}f_k$,
$\Delta_t=b_{t+1}-b_t$, and $d_k=s_k+\alpha\tau(s_k)$. In this subsection
$d_k$ is nondecreasing. Let $S^{\le r}$ be the seller distribution conditioned on the first $r$
types, so its probabilities are $f_k/F_r$ for $k\le r$, and
$\mathcal X_r$ the arrays in $[0,1]^{r\times n}$ nonincreasing in seller
index and nondecreasing in buyer index, extended by zero outside the cap.
For $z\in\mathcal X_r$ write
\begin{equation}\label{eq:type-contribution}
 \mathcal D_\ell^r(z)=\sum_{k\le r}f_k((1+\alpha)b_\ell-d_k)z_{k\ell}
       -\alpha\sum_{t<\ell}\Delta_t\sum_{k\le r}f_kz_{kt}.
\end{equation}
Expanding the buyer envelope gives, for every positive buyer probability
vector $g$,
\begin{equation}\label{eq:local-envelope}
 \sum_\ell g_\ell\mathcal D_\ell^r(z)
       =F_r\mathcal L_\alpha(z;g,S^{\le r}).
\end{equation}
Here $\mathcal D_\ell^r$ is the fixed-cap objective
$C^a(b_\ell)$ from \cref{eq:fixed-cap-lagrangian}, written with discrete
virtual costs. Conditioning on an initial set leaves the   {included}  virtual
costs unchanged, since the cumulative probabilities and point masses are
scaled by the same factor. Multiplication by $F_r$ returns the objective
to the original probability scale.

The next theorem exchanges the order of two choices: instead of selecting
a rule for each buyer prior, we obtain one rule that meets a target for
every buyer type. Moreover, the rules can be chosen consistently across
all seller caps. The numbers $C_\ell^r$ below are target lower bounds,
not the local objectives $C^a(b)$.

\begin{theorem}[Bilateral Bounds and Cap Monotonicity]\label{thm:compatible-targets}
Let $C_\ell^0=0$ and $0\le C_\ell^r\le C_\ell^{r+1}$ for every $\ell,r$.
The following statements are equivalent:
\begin{enumerate}[label=(\roman*),leftmargin=2em,itemsep=1pt]
\item For every nonempty cap and every positive probability vector $g$,
$F_rV_\alpha(g,S^{\le r})\ge\sum_\ell g_\ell C_\ell^r$.
\item There are $z^r\in\mathcal X_r$ with
$\mathcal D_\ell^r(z^r)\ge C_\ell^r$ for all $r,\ell$ and
$z^r\le z^{r'}$ coordinatewise whenever $r\le r'$.
\end{enumerate}
For rational data and targets, such a family is constructible in polynomial
time in the explicit local input and output size.
\end{theorem}

\begin{proof}[Proof of \cref{thm:compatible-targets}]
At a fixed cap, \cref{lem:exact-ironed-optimizer} supplies a pointwise
bilateral optimizer. Let $\Delta_n$ denote the probability simplex on the
$n$ buyer types, including zero probabilities. Thus (i), continuity at zero
buyer masses, and the finite minimax theorem imply
\[
 0\le\min_{g\in\Delta_n}\max_{z\in\mathcal X_r}
       \sum_\ell g_\ell(\mathcal D_\ell^r(z)-C_\ell^r)
 =\max_{z\in\mathcal X_r}\min_\ell(\mathcal D_\ell^r(z)-C_\ell^r).
\]
This gives one reference array satisfying every type inequality for that cap.
The minimax domains are compact and convex, and the displayed expression is
bilinear; no boundary prior is assigned a virtual value by division by zero.

For mass $\mu\in[0,F_r]$, fill the cheapest seller types:
\[
 T_k^r(\mu)=\min\{1,\max\{0,(\mu-F_{k-1})/f_k\}\}\quad(k\le r),
 \qquad T_k^r(\mu)=0\quad(k>r).
\]
Let $A_\ell^r(\mu)=\sum_{k\le r}f_k((1+\alpha)b_\ell-d_k)T_k^r(\mu)$.
Since $d_k$ increases, this column maximizes the current contribution at its
prescribed mass, even if some coefficients are negative. Starting from
$\mu_0^r=0$, choose
\begin{equation}\label{eq:greedy-mass}
 \mu_\ell^r=\min\left\{\mu\in[\mu_{\ell-1}^r,F_r]:
 A_\ell^r(\mu)\ge C_\ell^r+\alpha\sum_{t<\ell}\Delta_t\mu_t^r\right\}.
\end{equation}
Inductively the earlier chosen masses do not exceed the reference masses.
The reference's current mass is at least the lower endpoint, satisfies the
target and the larger reference rents, and benefits from cheapest-first
rearrangement. It is a feasible test mass, proving that the minimum exists.

Compare $r<r'$ by induction on $\ell$. If $\mu_\ell^{r'}<\mu_\ell^r$, that
mass lies below $F_r$, where both cheapest-first objectives agree. The
target for the larger cap and its earlier rents are no smaller. Moreover
$\mu_\ell^{r'}\ge\mu_{\ell-1}^{r'}\ge\mu_{\ell-1}^r$. Hence that mass would
satisfy the test for the smaller cap, contradicting minimality. The threshold columns
$z_{k\ell}^r=T_k^r(\mu_\ell^r)$ give both type monotonicity and cap-monotonicity.
They satisfy the inequalities by construction. Conversely (ii) implies (i)
by \cref{eq:local-envelope}. Linear programming finds reference arrays;
each first crossing scans at most $r$ affine pieces. Sequential rational
operations have polynomial bit length in the finite local instance.
\end{proof}

\Cref{cor:local-coefficient} shows that the largest uniform factor over these
buyer distributions and seller caps is attained and computable by a
polynomial-size LP. This LP does not compute a whole market's second best
or characterize every optimal matching mechanism. When the separate caps
admit different bounds, \cref{ct:app:targets} gives stronger cumulative
targets and the corresponding market-level guarantee.

\paragraph{Applying the Composition Theorem.}
The constructed family now meets all the monotonicity requirements of
\cref{thm:edge-composition}. To identify its objective, recall that for finite seller priors, the normalized receipt of a monotone column $z$ is
$s_kz_k+\sum_{h>k}(s_h-s_{h-1})z_h$. Its expectation is
$\sum_k f_k\tau(s_k)z_k$. Consequently \cref{eq:type-contribution} is
exactly the typewise Lagrangian expression in \cref{eq:fixed-cap-lagrangian}. Use
\cref{thm:edge-composition} with the family from
\cref{thm:compatible-targets}. Writing $r_\ell$ for the number of types in the first-best cap and
$\mu_t^r=\sum_kf_{jk}z_{kt}^r$, its only correction is
\begin{equation}\label{eq:rent-correction}
 \alpha\sum_{t<\ell}(b_i^{t+1}-b_i^t)
 (\mu_t^{r_\ell}-\mu_t^{r_t})\ge0.
\end{equation}
The general target conclusion is
\begin{equation}\label{eq:cap-target-bound}
 V_\alpha(\mathcal I)\ge
 \sum_{e=(i,j)}\E_\omega\Bigl[\sum_\ell g_{i\ell}C_{e,\ell}^{r_{e,\omega}(\ell)}\Bigr].
\end{equation}
At targets $C_{e,\ell}^r=\beta\sum_{k\le r}f_{jk}(b_i^\ell-s_j^k)_+$,
this is $\beta\FB$. No additional cap-probability factor is introduced.

\subsection{Regularizing Sellers Before Choosing Caps}
\label{sec:global-ironing-v2}
The local construction assumed nondecreasing seller scores. To handle
arbitrary sellers, we transform each full seller distribution once, before
conditioning on any initial set of costs. This keeps the scores consistent
across caps and preserves the probabilities of the ordered types.
The next proposition states the two comparisons needed for the reduction.

\begin{proposition}[Seller Regularization]\label{prop:regularization-main}\label{prop:global-ironing-package}
At each fixed $\alpha$, each finite seller prior admits an
$\alpha$-weakly regular replacement with the same masses and number of
ordered types. Applying these replacements independently gives a market
$\widehat{\mathcal I}^{\alpha}$ satisfying
\begin{equation}\label{eq:regularization-summary}
 V_\alpha(\mathcal I)=V_\alpha(\widehat{\mathcal I}^{\alpha}),\qquad
 \FB(\mathcal I)\le\FB(\widehat{\mathcal I}^{\alpha}).
\end{equation}
The replacement depends only on the full seller prior and $\alpha$.
Under the coupling by type indices, its combined seller scores $s+\alpha\tau(s)$
are the original ironed scores. The canonical maximizing matching
is therefore the same at every index profile, for every buyer prior and
compatibility graph with the stated feasibility constraints.
\end{proposition}

The construction and its proof are given in \cref{app:regularization-proof}.
Its key idea is to iron the full seller score sequence and then reconstruct
actual costs whose combined scores equal those ironed scores. The
reconstructed costs are strictly increasing and   {preserve}  the original masses.
Thus the two markets have the same optimized ironed objective at every
profile of type indices, which preserves $V_\alpha$.

To compare first-best gains, fix every report except one seller's cost $c$.
If no feasible matching uses that seller, replacing her prior has no effect.
Otherwise, the best matching without that seller has some value $A$, and the best
matching using her has value $B-c$. Their maximum is
$A+(B-A-c)_+$. The reconstructed prior weakly increases the expectation of
$(\kappa-c)_+$ for every threshold $\kappa$, so replacing sellers one at a
time can only increase $\FB$. The transformation may move the   {support
points}; this is why restricted applications must allow those   {support points}
 to change, as support-size bounds do.

The reduction now follows by applying the local bounds in the transformed
market and returning through equality of optimized values. No allocation is
reinterpreted at the original seller costs.

\begin{proof}[Proof of \cref{thm:bilateral-to-matching}]
Apply \cref{prop:regularization-main}. Every normalized prefix of a transformed
seller is weakly regular, so the uniform premise supplies all local
inequalities for all buyer distributions on the fixed support. Use \cref{thm:compatible-targets,eq:cap-target-bound}
in the transformed market and then \cref{eq:regularization-summary}.
The original market's own optimizer from \cref{lem:exact-ironed-optimizer}
attains the bound using a deterministic allocation with positive-surplus
trades. Markets with no positive surplus use no trade.
\end{proof}

This completes the finite reduction. We next turn to the MHR application,
where a direct continuous construction supplies the local rules for
the same composition theorem.

\section{Sharp Bounds for MHR Buyers}\label{sec:mhr-overview}\label{sec:sharp-mhr}
We first apply the composition theorem to buyers with monotone hazard rate
(MHR). We construct bilateral rules directly on continuous report intervals,
then prove the same sharp efficiency bound for bilateral trade and matching
markets.

For the   MHR priors considered here, the continuous part has density $f$
and distribution function $F$ with nondecreasing hazard
$f(b)/(1-F(b))$; an atom at the upper endpoint is allowed.  To state the exact constant,
for $0<t<1$ define
\begin{equation}\label{mhr:eq:M}
 M_t(z)=t\int_0^1 u^{t-1}e^{zu}\,\dd u
       =\sum_{n=0}^\infty\frac{t}{t+n}\frac{z^n}{n!}.
\end{equation}
Let $z_t$ be the unique root in $(1-t,2(1-t))$ of
\begin{equation}\label{mhr:eq:balance-root}
 (z_t+2t-1)M_t(z_t)=te^{z_t},
\end{equation}
and set
\begin{equation}\label{mhr:eq:sharp-curve}
 c_0=1,\qquad c_\alpha=M_t(z_t)^{-1},\quad t=(1+\alpha)^{-1}
 \quad(\alpha>0).
\end{equation}
Thus computing the bound for a given multiplier requires solving one scalar
equation. We prove below that its root exists and is unique.

\begin{theorem}[Sharp MHR Bounds]\label{thm:sharp-mhr}
For every $\alpha\ge0$, the infimum of $V_\alpha/\FB$ over independent
bounded bilateral priors with an MHR buyer and an arbitrary seller equals
$c_\alpha$. The same infimum holds over bounded downward-closed matching
markets with all buyers MHR and arbitrary sellers. In both classes,
\begin{equation}\label{mhr:eq:gamma}
 \inf_{\FB>0}\frac{\SB}{\FB}
 =\gamma_{\rm MHR}:=\min_{0<t<1}\frac{1}{M_t(z_t)}
 \simeq 0.72490721.
\end{equation}
The minimum is attained in the interior. Requiring bounded absolutely
continuous buyers leaves all these infima unchanged. The lower bounds extend
to a common finite lower bound and finite expected first best.
\end{theorem}

The theorem gives the optimal bound at each budget multiplier and the
resulting efficiency guarantee. The exact constant is
\eqref{mhr:eq:gamma}; the decimal is illustrative, and uniqueness of the
minimizing $t$ is not asserted.

We first bound the bilateral Lagrangian separately for \emph{every} report
of the unrestricted agent. We then prove cap-monotonicity and apply
the composition theorem, before constructing tight bilateral instances.
At $\alpha=0$, first best is optimal; the construction treats $\alpha>0$.

It is convenient to exchange the two sides before constructing the rule.
Reflect all types around a common upper bound $U$: a buyer value $b$ becomes
seller cost $U-b$, and seller cost $c$ becomes buyer value $U-c$. Surplus is
unchanged. Moreover,  if  an MHR buyer   is reflected into the seller cost $U-B$, then
the logarithm of the reflected seller's  cumulative distribution function
(CDF)   is concave on its positive region, since
 $\Pr[U-B\le y]=\Pr[B\ge U-y]$. Such a CDF is locally absolutely continuous
in the interior of its positive region and can have an atom only at its
lower support endpoint; deterministic distributions are also allowed.
We work in these reflected coordinates until applying the composition theorem.
The model's assumptions on independence, bounded reports, and transfers
remain in force; $V_\alpha$ is a supremum over BIC/IIR mechanisms without
a budget constraint.

\subsection{ { An Exponential Comparison for   Reflected MHR  CDFs } }\label{mhr:sec:crossing}
We first justify the constant and prove the comparison inequality needed
below. The reference case has   affine $\log F$; concavity of $\log F$  extends the
bound to every permitted CDF.
Write $\eta=1-t$ and $A=1/t=1+\alpha$. Integration by parts gives
\begin{equation}\label{mhr:eq:M-ode}
 zM_t'(z)+tM_t(z)=t e^z.
\end{equation}
For $Q_t(z)=(z+2t-1)M_t(z)-t e^z$,
\[
 Q_t'(z)=\eta M_t(z)+(2t-1)M_t'(z)>0,
\]
since $0<M_t'(z)<M_t(z)$; when $t<1/2$, the right side is greater than
$tM_t(z)$. Also $Q_t(\eta)=t(M_t(\eta)-e^\eta)<0$. At $z=2\eta$,
\[
 Q_t(2\eta)=\eta\bigl(M_t(2\eta)-2M_t'(2\eta)\bigr)>0.
\]
For the last sign, pair $u$ and $1-u$ in the integral. For $0<u<1/2$, their
weight ratio is
\[
 \frac{u^{t-1}e^{2\eta u}}{(1-u)^{t-1}e^{2\eta(1-u)}}
 =\exp\!\left(\eta\left[\log\frac{1-u}{u}-2(1-2u)\right]\right)>1.
\]
Hence the weighted mean of $u$ is below $1/2$. This proves existence and
uniqueness of $z_t$ in the stated interval. Fix $\alpha>0$ from now on and write
$z=z_t$, $\beta=c_\alpha$. We will use
\begin{equation}\label{mhr:eq:balance}
 \beta M_t(z)=1,\qquad t\beta e^z=z+2t-1.
\end{equation}

The next lemma compares a concave function $h$ with the line $zu$ after
matching their weighted exponential averages. It gives the derivative
inequality needed to bound information rents.
\begin{lemma}\label{mhr:lem:crossing}
Let   $h$  be nondecreasing, concave, and Lipschitz  on $[0,1]$, with $h(0)=0$.
Suppose
\[
 \beta t\int_0^1u^{t-1}e^{h(u)}\dd u=1.
\]
Then
\begin{equation}\label{mhr:eq:shape-ineq}
 \beta t\int_0^1u^{t-1}(1-u)h'(u)e^{h(u)}\dd u\ge1-t.
\end{equation}
\end{lemma}
\begin{proof}
Put $g(u)=e^{h(u)}-e^{zu}$. The normalization and \eqref{mhr:eq:balance} give
$\int_0^1u^{t-1}g(u)\dd u=0$. Concavity makes $h(u)/u$ nonincreasing, so
$g$ changes sign at most once, from nonnegative to nonpositive. Multiplication
by the decreasing function $1/u$ therefore gives
\begin{equation}\label{mhr:eq:cross-moment}
 \int_0^1u^{t-2}g(u)\dd u\ge0.
\end{equation}
Explicitly, if $a$ is a crossing point, subtract
$a^{-1}\int u^{t-1}g=0$; the remaining integrand
$u^{t-1}(u^{-1}-a^{-1})g(u)$ is nonnegative. If there is no nontrivial crossing,
the normalization forces $g=0$ almost everywhere. The integral is finite
because $g(u)=O(u)$ at zero.

Integration by parts, with both boundary terms vanishing, now gives
\[
 \int_0^1u^{t-1}(1-u)g'(u)\dd u
 =(1-t)\int_0^1u^{t-2}g(u)\dd u
   +t\int_0^1u^{t-1}g(u)\dd u\ge0.
\]
Thus the left side of \eqref{mhr:eq:shape-ineq} is at least its value for $h(u)=zu$.
That value is
\[
 \beta z(M_t(z)-M_t'(z))
 =z+t-t\beta e^z=1-t,
\]
by \eqref{mhr:eq:M-ode} and \eqref{mhr:eq:balance}.
\end{proof}

The normalization is essential: the comparison is not a pointwise tangent
bound. It compares the weighted derivatives after the corresponding weighted averages have
been matched. This is precisely the derivative quantity that appears in the
information-rent calculation below.

\subsection{Typewise Bilateral Rules and Bounds}\label{mhr:sec:typewise-bound}
We now construct a bilateral rule and prove its bound separately for each
buyer value. This typewise guarantee is what the composition theorem needs;
it also allows the buyer's distribution to be arbitrary.

Translate the lower support endpoint of a reflected seller to zero. Its CDF is
$F$, supported on $[0,H]$, with $F(v)=1$ for $v\ge H$. The extension
$\log F=0$ above $H$ remains concave. Write $p=F(0)$ for the probability
of the bottom atom, and let $T(v)=\E[(v-Y)_+]$ be first-best surplus
conditional on buyer value $v$. In particular,
\[
 T(v)=\int_0^v F(w)\dd w\quad(v\ge0),\qquad T(v)=0\quad(v\le0).
\]
Deterministic $Y=0$ is allowed. For a buyer value $v>0$, define the weighted
average of the CDF between a candidate cutoff $s$ and $v$ by
\[
 K_F(s,v)=t\int_0^1u^{t-1}F(s+u(v-s))\dd u,
 \qquad 0\le s\le v.
\]
Let $m(v)$ denote the total probability of trade over the seller prior,
and let $s(v)$ denote the largest selected seller cost. We choose them by
\begin{equation}\label{mhr:eq:construction}
 \begin{cases}
 m(v)=\beta K_F(0,v),\quad s(v)=0,
       &\beta K_F(0,v)\le p,\\[1mm]
 F(s(v))=\beta K_F(s(v),v),\quad m(v)=F(s(v)),
       &\beta K_F(0,v)>p.
 \end{cases}
\end{equation}
In the second case choose $s(v)\in(0,\min\{v,H\})$.
Set $m(v)=0$ for $v\le0$. In the first case, trade only at the bottom atom,
with probability $m(v)/p$ there; we call this the \emph{atom case}.
In the second, trade exactly at seller costs $Y\le s(v)$; we call this the
\emph{interior case}. When $p=0$, the atom case is absent for $v>0$.
Thus the rule trades only when $Y<v$. The weighted-average identity will
give the required surplus bound.

\paragraph{Reports Outside the Seller Support.}
Translation need not make zero the lowest \emph{allowed report}.
For $v>0$, extend the allocation constantly to seller reports $y<0$:
its trade probability is $m(v)/p$ in the atom case when $p>0$, and one in the
interior case. Reports above the selected cutoff receive zero allocation.
For $v\le0$ set the allocation to zero at every seller report. This extension
is monotone on the full report domain and trades only at positive surplus.
It changes neither the seller mass nor receipts at supported costs: the
seller envelope integrates from the reported cost upward. The buyer envelope
below zero also vanishes. Thus the payment identities below apply on the full report domains,
including reports that have zero probability under the prior.

\paragraph{Existence, Monotonicity, and Regularity.}
We next verify that the cutoff is well-defined and that higher buyer values
and lower seller costs make trade weakly more likely. For fixed $v$, define
\[
 R_F(s,v)=\frac{K_F(s,v)}{F(s)}.
\]
On the positive, strictly increasing part of $F$, this is continuous and strictly
decreasing in $s$. Indeed, if $r=(\log F)'$, concavity gives, almost everywhere,
\[
 \frac{\partial}{\partial s}\log\frac{F(s+u(v-s))}{F(s)}
 =(1-u)r(s+u(v-s))-r(s)\le-u r(s)<0.
\]
At $s=v$ the ratio is one; at $s=H<v$ it is also one. In the second case of
\eqref{mhr:eq:construction}, its limiting value at zero is greater than $1/\beta$,
including $+\infty$ when $p=0$. This proves existence and uniqueness of the root.
The root satisfies $m(v)=\beta K_F(s(v),v)\le\beta<1$.

Increasing $v$ increases $R_F(s,v)$ at a fixed $s$, hence increases the cutoff.
Shifting both arguments to the right, while keeping $v-s$ fixed, decreases
$R_F$ by concavity of $\log F$. Consequently, in the interior case,
\[
 0\le s(v_2)-s(v_1)\le v_2-v_1 \qquad(v_2\ge v_1).
\]
Thus the cutoff is locally Lipschitz, and so is $m$ away from zero.
The bottom-atom mass is also nondecreasing. At the transition the mass tends to
$p$ and the cutoff to zero, so the allocation is monotone in both types.
For completeness, the seller's continuous virtual cost is $y+1/r(y)$,
where $r=(\log F)'=F'/F$ is the reverse hazard rate. Because $r$ is
nonincreasing, the seller's Lagrangian score $Ay+\alpha/r(y)$ is
nondecreasing; the bottom-atom score is zero. Thus the reflected seller is
regular. No regularity of the opposite-side prior is needed.

\paragraph{The Information-Rent Inequality.}
We now show that, after accounting for incentive payments, the Lagrangian
grows at least as fast as $\beta$ times conditional first-best surplus.
Let $D(m)$ be the allocated seller cost plus $\alpha$ times the seller's
threshold-payment receipt, averaged over the seller prior. In the atom case
$D(m)=0$; for $m=F(s)>p$, integration by parts gives
\[
 D(F(s))=\int_{[0,s]}y\dd F(y)+\alpha sF(s)
       =A sF(s)-T(s).
\]
Define
\begin{equation}\label{mhr:eq:C}
 I(v)=\int_0^v m(w)\dd w,\qquad
 C(v)=Avm(v)-D(m(v))-\alpha I(v).
\end{equation}
The desired bound, which does not depend on the buyer prior, is
\begin{equation}\label{mhr:eq:typewise-bound}
 C(v)\ge\beta T(v)\qquad\text{for every }v.
\end{equation}
It suffices to prove $C'(v)\ge\beta F(v)$ almost everywhere and then
integrate from zero. In the atom case,
$m(v)=\beta t v^{-t}\int_0^v w^{t-1}F(w)\dd w$, so
\[
 C'(v)=m(v)+Avm'(v)=\beta F(v)
\]
at almost every $v$.

In the interior case, write $s=s(v)$ and $d=v-s>0$, and put
\[
 h(u)=\log F(s+du)-\log F(s),\qquad
 \dd\mu(u)=\beta t u^{t-1}e^{h(u)}\dd u.
\]
Equation~\eqref{mhr:eq:construction} says $\mu$ is a probability measure. Differentiating
that equation along $v$ gives
\begin{equation}\label{mhr:eq:implicit-derivative}
 r(s)s'=\E_\mu\bigl[r(s+du)\{u+(1-u)s'\}\bigr].
\end{equation}
An integration by parts in $u$ gives
\begin{equation}\label{mhr:eq:moment-identity}
 d\E_\mu[u r(s+du)]=t\left(\beta\frac{F(v)}{F(s)}-1\right).
\end{equation}
Since $m'=m r(s)s'$ and $A=1/t$, differentiating \eqref{mhr:eq:C} first gives
\[
 C'(v)=m(v)+\frac{m(v)s'(v)}{t}\bigl(d\,r(s)-(1-t)\bigr).
\]
Multiply \eqref{mhr:eq:implicit-derivative} by $d$ and substitute
\eqref{mhr:eq:moment-identity}; the terms $m-\beta F(v)$ cancel. This gives
\begin{equation}\label{mhr:eq:key-derivative}
 \begin{split}
 C'(v)-\beta F(v)
 &=\frac{m(v)s'(v)}t
    \left(\E_\mu[(1-u)h'(u)]-(1-t)\right)\\
 &\ge0.
 \end{split}
\end{equation}
The last sign is Lemma~\ref{mhr:lem:crossing}. In other words, an affine
$\log F$ is the worst case for this derivative
bound. Any additional concavity contributes a nonnegative term once the
cutoff satisfies the weighted-average identity.

All differentiations hold almost everywhere. On compact interior ranges,
$\log F$ is Lipschitz and $r(s)>0$; the implicit relation has a nonzero derivative
in $s$, and the preceding shift comparison gives the required local absolute
continuity. A corner of $\log F$ or the flattening at $H$ causes no jump in $C$.
The two cases join continuously, $C(0)=T(0)=0$, and the possible jump of $m$ at
zero is multiplied by $v=0$. Integrating \eqref{mhr:eq:key-derivative} therefore
proves \eqref{mhr:eq:typewise-bound}. The case $p=0$ uses the same proof on each compact
range $v>0$, then lets its left endpoint tend to zero; both $C(v)$ and $T(v)$
tend to zero. Deterministic $Y=0$ has $m(v)=\beta$ for $v>0$ and satisfies equality.

\paragraph{Interpretation as a Bilateral Mechanism.}
The buyer's interim envelope payment is $v m(v)-I(v)$, where $I(v)$ is
the information rent needed for truthful reporting. In the interior
case each selected seller receives the threshold payment $s(v)$; in the atom
case the receipt is zero.
Substituting these payments shows that $C(v)$ is exactly the conditional
value of $\GFT+\alpha\operatorname{Rev}$. Monotonicity gives incentive
compatibility for both agents, and averaging
\eqref{mhr:eq:typewise-bound} over any independent buyer prior yields
$V_\alpha\ge\beta\FB$. That prior may have atoms or singular parts.
If its report interval begins below zero, the allocation vanishes there,
so the same envelope formula applies. We have therefore proved the bilateral
lower bound with a rule suitable for the composition argument.

\subsection{From Bilateral Rules to Matching Markets}\label{mhr:sec:cap-composition}
To compose the bilateral rules, allowing additional seller costs must never
reduce trade at a fixed report pair. We verify this property, apply the
composition theorem, and recover budget feasibility from the multiplier bounds.

\paragraph{Monotonicity Across Seller Caps.}
As in the composition theorem,   {keep}  the original seller probabilities
when restricting the seller to an initial segment $Y\le a$. Its cumulative
mass is
\[
 F_a(y)=\Pr[Y\le y,\,Y\le a]=F(\min\{y,a\}),\qquad y\ge0.
\]
Thus   $\log F_a$ is concave on its positive region and $F_a$  has total mass
$F(a)$; dividing by $F(a)$   preserves this concavity property. The construction is
homogeneous: multiplying $F$ by a positive constant multiplies $m$, $C$, and $T$ by that constant and leaves its cutoff
unchanged. Apply \eqref{mhr:eq:construction} to $F_a$ and write the resulting mass
as $m^a(v)$.

For $a\le b$, one has $F_a\le F_b$. The allocated mass in the atom case
consequently increases.
If the smaller cap has an interior cutoff $s<a$, its denominator $F(s)$ is
unchanged by either cap, while the numerator $K_{F_b}(s,v)$ is at least
$K_{F_a}(s,v)$. The decreasing-root characterization forces the cutoff for
the larger cap, and hence its allocated mass, to be no smaller. Thus
\begin{equation}\label{mhr:eq:cap-nesting}
 m^a(v)\le m^b(v)\quad(a\le b),
\end{equation}
and the trade probability at every fixed pair of reports is nondecreasing
as the seller cap expands. This is exactly cap-monotonicity.
When the cap consists only of the bottom atom, the allocated
mass is $\beta p$ for $v>0$.

Use the extension to reports outside the support for each positive-mass
cap, and assign zero allocation to every zero-mass cap.
At an interior endpoint, open and closed caps have the same mass,
since there is no interior atom, and the selected cutoff lies strictly below
that endpoint. At the bottom atom, excluding the endpoint gives zero
allocation, whereas including it gives the atom case. These conventions
preserve cap-monotonicity on the full report domain.

The integral $K_{F_a}(s,v)$ is Borel in its arguments. In the interior
case its strictly decreasing ratio characterizes the cutoff by rational
level tests; in the atom case the mass is an explicit integral.
Thus the cutoffs, allocations and targets are jointly Borel measurable.
The first-best cap endpoint is measurable as well: it is obtained by
comparing finitely many affine matching weights, with the fixed tie order
deciding inclusion. These facts verify the measurability hypothesis of
\cref{thm:edge-composition}.

\paragraph{Composition Across First-Best Edges.}
 { Reflect each original MHR buyer into a   seller whose CDF has concave logarithm and translate}
 its lower support to zero. The construction above gives the measurable
cap-monotone family required by \cref{thm:edge-composition}. Its conditional
Lagrangian is \cref{mhr:eq:C}, so \cref{mhr:eq:typewise-bound} supplies the
target for every seller cap and opposite report. The theorem applies with the fixed
first-best tie order; the cap is empty whenever the buyer cannot
trade in first best. In its interval form,
\cref{eq:unified-rent-correction} is precisely
\begin{equation}\label{mhr:eq:cap-rent}
 \alpha\int_L^v\bigl(m^{a(v)}(w)-m^{a(w)}(w)\bigr)\,\dd w\ge0.
\end{equation}
The integrand is nonnegative because a lower buyer report selects a smaller
seller cap. Consequently, the local bound gives
$V_\alpha\ge c_\alpha\FB$ in the reflected matching market.

To transfer this result back to the original market, recall that reflection
preserves surplus. It also preserves incentive compatibility and revenue.
Write all transfers as payments \emph{to} the mechanism, including the
seller's transfer. Under $(b',s')=(U-c,U-b)$, reflected transfers
$(p'_B,p'_S)$ become
$p_B=Ux+p'_S$ and $p_S=p'_B-Ux$.
These identities preserve both agents' utilities and total revenue, and
apply to service probabilities in a matching as well. Equivalently, when
seller transfers are written as positive \emph{receipts}, the formulas are
$P_B=Ux-P'_S$ and $P_S=Ux-P'_B$.
Thus the same fixed-multiplier bound holds for the original independent priors.

\paragraph{Recovering Budget Feasibility.}
It remains to obtain nonnegative expected revenue. We prove this directly
for continuous priors, without assuming that the finite linear program's
attainment argument applies.

Let $\mathcal K$ be the closure of the convex hull of the gain--revenue pairs
of the mechanisms constructed above, no trade, and a fixed positive-revenue
posted-price rule. Such a rule exists whenever $\FB>0$: on some allowed
edge, buying at one price and selling at a strictly higher price results in
trade with positive probability. Normalized threshold or interim-envelope
payments and bounded types make $\mathcal K$ compact. For all $\alpha\ge0$,
\[
 \max_{(G,R)\in\mathcal K}(G+\alpha R)\ge c_\alpha\FB
          \ge\gamma_{\mathrm{MHR}}\FB.
\]
If $\mathcal K$ were disjoint from $\{G\ge\gamma_{\mathrm{MHR}}\FB,R\ge0\}$,
strict separation would give nonnegative coefficients $(u,v)$ with
$\max_{\mathcal K}(uG+vR)<u\gamma_{\mathrm{MHR}}\FB$. If $u=0$, the positive
revenue point contradicts this; if $u>0$, divide by $u$ and use the displayed
inequality. Thus the sets intersect. Approximating by actual mixtures and
repairing any vanishing expected deficit with the fixed positive-revenue rule
gives budget-feasible mechanisms whose gains approach
$\gamma_{\mathrm{MHR}}\FB$. This proves the desired second-best guarantee,
with $\SB$ interpreted as a supremum.

The budget conversion in \cref{app:sbb-conversion} uses only independence
and integrable interim payments, so it also applies to these bounded
continuous priors. It gives profilewise strong budget balance while
preserving BIC and IIR.

\paragraph{Unbounded Upper Supports.}
The bounded result also gives the stated extension when first-best surplus
has finite expectation. Translate a common finite lower bound to zero and
clip all types at $K$.
  {Clipping preserves the MHR property of buyers.}  The bounded
rules have positive-surplus support, so \cref{lem:clipping-extension}
extends them with unchanged payments and no smaller gains;
\cref{lem:clipping-integrability} ensures integrability. By
\cref{lem:capped-fb-convergence}, clipped first best increases to true first
best. Taking limits proves the second-best bound. The incentive and
payment calculation in the proof of \cref{lem:clipping-extension} does not
use budget nonnegativity. It therefore also extends each fixed-multiplier
witness with unchanged revenue and no smaller GFT, proving the $V_\alpha$
bound.

\subsection{Tight Bilateral Instances}\label{mhr:sec:upper}
We now show tightness at every fixed multiplier. Bilateral instances
suffice, since a single-edge market is also an allowed matching market. Fix $t\in(0,1)$, let
$\eta=1-t$, $\alpha=\eta/t$, and choose $k\in(0,1)$ and
\[
 z=\eta(1+k),\qquad 0<\delta<z<H.
\]
In reflected coordinates let $Y$ have CDF $F_H(y)=e^{y-H}$ on $[0,H]$, including
the atom $e^{-H}$ at zero. Independently, give $X\in[\delta,H]$ the   {tail probability}
 function
\begin{equation}\label{mhr:eq:adversary}
 \Pr[X\ge x]=
 \begin{cases}
 1,&x\le\delta,\\
 (\delta/x)^\eta,&\delta<x\le z,\\
 (\delta/z)^\eta e^{-(x-z)/k},&z<x\le H,\\
 0,&x>H.
 \end{cases}
\end{equation}
There is an atom at $H$, but not at $\delta$ or $z$. The original buyer is
$B=H-Y=\min\{\operatorname{Exp}(1),H\}$, which is MHR, and the original seller has cost
$S=H-X\in[0,H-\delta]$.

This family makes the comparison step tight: $\log F_H$ is affine, and
the buyer prior is chosen to have a simple Lagrangian score. Write these
scores as $d_Y(y)$ and $a_X(x)$. The reflected seller score $d_Y$ is zero at
its bottom atom and $Ay+\alpha$ in the interior. The reflected buyer score
$a_X$ is zero on $(\delta,z)$, is $Ax-\alpha k$ on $(z,H)$, and is $AH$ at
its top atom. Both score functions are nondecreasing. Trading exactly when
the buyer score exceeds the seller score is therefore monotone and maximizes
the Lagrangian. For $z<x<H$, its conditional value is
\[
 \E_Y[(a_X(x)-d_Y(Y))_+]=e^{-H}(Ae^{x-z}-1),
\]
and at $x=H$ it is $e^{-H}(Ae^{H-\eta}-1)$.
Set $P=(\delta/z)^\eta$ and $Q=e^{-(1/k-1)(H-z)}$. Integration gives
\begin{align}
 \frac{V_\alpha}{e^{-H}P}
 &=\frac{A}{1-k}(1-Q)-1+Ae^{z-\eta}Q,\label{mhr:eq:upper-N}\\
 \frac{\FB}{e^{-H}P}
 &=\frac{z}t M_t(z)
   +z^\eta\delta^t\left(\frac{e^\delta-1}{\delta}
                         -\frac{M_t(\delta)}t\right)
   +e^z\frac{k}{1-k}(1-Q).\label{mhr:eq:upper-D}
\end{align}
For the denominator, use
$\E[(X-Y)_+]=e^{-H}\E[e^X-1]$ and the tail integral. On $[\delta,z]$,
\[
 \frac{\E[e^X-1]\text{ contribution}}{P}
 =z^\eta\!\int_\delta^z e^x x^{t-1}\dd x
   +z^\eta\delta^{-\eta}(e^\delta-1),
\]
which is the first two terms of \eqref{mhr:eq:upper-D}. The remaining exponential
segment is a geometric integral. For the numerator, integrate the conditional
value against density $P e^{-(x-z)/k}/k$ and add the top atom.

First let $\delta\downarrow0$, then $H\to\infty$. Equations
\eqref{mhr:eq:upper-N}--\eqref{mhr:eq:upper-D} yield
\begin{equation}\label{mhr:eq:upper-ratio}
 \frac{V_\alpha}{\FB}\longrightarrow
 \frac{1/t-1+k}{(1-k)(z/t)M_t(z)+k e^z}.
\end{equation}
Now take $z=z_t$ and $k=z_t/(1-t)-1\in(0,1)$.
The balance equation \eqref{mhr:eq:balance-root} makes the last ratio exactly
$1/M_t(z_t)=c_\alpha$. Thus every fixed-multiplier lower bound is tight in
bounded bilateral instances. At $\alpha=0$, $V_0=\FB$.

The root is continuous in $t$ by its strict $z$ derivative. As $t\to1$,
$z_t\to0$; as $t\to0$, $z_t$ remains bounded and $M_t(z_t)\to1$ by its positive
series. Hence $c_\alpha\to1$ at both ends, while $c_\alpha<1$ in the interior.
The minimum in \eqref{mhr:eq:gamma} is attained. Fix a minimizing multiplier in the
upper family and use $\SB\le V_\alpha$ to match the second-best lower bound.
A bilateral instance is itself a matching market, proving the two equalities
in Theorem~\ref{thm:sharp-mhr}.

\paragraph{Removing the Endpoint Atom.}
The tight family uses an atom at the buyer's upper endpoint. We can replace
it by a short continuous tail while preserving MHR and changing the ratio
arbitrarily little.
For a fixed $H,\delta$, replace the buyer atom $e^{-H}$ at $H$ by a uniform
conditional value on $[H,H+\epsilon]$, $0<\epsilon\le1$. The   {tail probability}  is $e^{-b}$
below $H$ and $e^{-H}(1-(b-H)/\epsilon)$ on the added interval. The hazard
jumps from one to $1/(H+\epsilon-b)$ and is nondecreasing. This buyer is bounded
and absolutely continuous. Below $H$ its Lagrangian score is unchanged; on the new
tail it differs from the old atom's score by at most $(1+\alpha)\epsilon$.
Both the old and new buyer scores, as well as the constructed seller score,
are nondecreasing, so the positive-part formula computes the optima. Coupling the old atom to the new tail
bounds the change of $V_\alpha$ by $(1+\alpha)e^{-H}\epsilon$, and the change
of $\FB$ by $e^{-H}\epsilon$. Letting $\epsilon\downarrow0$ for a fixed
positive-first-best instance preserves its ratio. This proves the atom-free
assertion of Theorem~\ref{thm:sharp-mhr}.

The next section applies the same composition theorem to finite supports,
where a cutoff recurrence replaces the analytic comparison used here.

\section{Sharp Bounds for Binary Buyers}\label{hs:sec:sharp}
We next replace the distributional assumptions of the MHR setting by a
bound on the number of possible values: each buyer has at most two types.
The finite-support reduction again makes the bilateral bound sufficient,
and it reveals exactly how the number of seller types affects efficiency.

Let $\mathfrak M_{2,m}$ be the class of markets with at most two types per buyer,
at most $m$ types per seller, independent finite priors, and any nonempty
downward-closed family of feasible matchings. Let $\mathfrak B_{2,m}$ be its
bilateral subclass. For the budget multiplier $\alpha\ge0$, write
$q=\alpha/(1+\alpha)$ and $D_m(q)=1+q-q^2-q^m$.

\begin{theorem}[Sharp Bounds for Binary Buyers]\label{thm:curve}
For every integer $m\ge1$ and every $\alpha\ge0$,
\begin{equation}\label{eq:exact-curve}
 \inf_{\mathcal I\in\mathfrak M_{2,m},\,\FB>0}\frac{V_\alpha(\mathcal I)}{\FB(\mathcal I)}
 =\inf_{\mathcal I\in\mathfrak B_{2,m},\,\FB>0}\frac{V_\alpha(\mathcal I)}{\FB(\mathcal I)}
 =D_m(q)^{-1}.
\end{equation}
Consequently the exact worst-case second-best ratio is
\begin{equation}\label{eq:gamma-definition}
 \gamma_m=\left[\max_{0\le q<1}D_m(q)\right]^{-1}.
\end{equation}
In particular,
\[
 \gamma_1=1,\qquad \gamma_2=\frac{8}{9},\qquad
 \gamma_3=\frac{27}{32},\qquad \gamma_m\downarrow\frac{4}{5}.
\]
All upper bounds are approached by finite rational bilateral instances in
$[0,1]$. The same statements hold after exchanging the two sides.
\end{theorem}

Thus competition among possible trading partners causes no further loss
in the worst-case guarantee. For example, two types on each side give the
ratio $8/9$, while allowing arbitrarily many seller types lowers the ratio
to $4/5$. We prove the bilateral lower bound first, transfer it to matching
markets, and then give bilateral instances approaching equality.

\subsection{The Bilateral Lower Bound and Its Matching Extension}\label{app:binary-lower-proof}
Fix a seller with $r\le m$ distinct costs $s_1<\cdots<s_r$, masses $f_k>0$, and
prefix probabilities $F_k=\sum_{h\le k}f_h$. Its discrete virtual costs satisfy
\[
 \tau(s_1)=s_1,\qquad
 \tau(s_k)=s_k+(s_k-s_{k-1})\frac{F_{k-1}}{f_k},\qquad
 \sum_{h\le k}f_h\tau(s_h)=F_ks_k.
\]
Regularity is not assumed. A seller cutoff $k$ trades exactly at the first
$k$ costs $s_1,\ldots,s_k$. Define the expected efficient surplus at buyer
value $v$ and the best Lagrangian value among these cutoffs by
\begin{align}
 H(v)&=\E[(v-S)_+],\label{hs:eq:H}\\
 K_\alpha(v)&=\max_{0\le k\le r}
 \left\{(1+\alpha)vF_k-\sum_{h\le k}f_hs_h-\alpha F_ks_k\right\}.
 \label{hs:eq:K}
\end{align}
The $k=0$ term is zero. For $k>0$, the expression in braces is the Lagrangian
contribution of the allocation that trades exactly on that seller prefix,
when the buyer-side coefficient is $(1+\alpha)v$. The following inequalities
relate this cutoff value to efficient surplus and to an increase in the
buyer's value. They will control the seller's contribution and the buyer's
information rent, respectively.

\begin{lemma}[Cutoff Bounds]\label{hs:lem:cutoff}
For every  $v$, every $\alpha\ge0$, and $q=\alpha/(1+\alpha)$,
\begin{equation}\label{hs:eq:cutoff-bound}
  H(v)\le(1-q^m)K_\alpha(v).
\end{equation}
Moreover, if $x<y$ and $F(x)=\Pr[S\le x]$, then
\begin{equation}\label{hs:eq:cutoff-slope}
  K_\alpha(y)\ge(1+\alpha)(y-x)F(x).
\end{equation}
\end{lemma}
\begin{proof}
Let $k_*$ be the number of costs strictly below $v$. If it is zero,
\cref{hs:eq:cutoff-bound} is immediate. Otherwise put
$A_k=\sum_{h\le k}f_h(v-s_h)$ for $0\le k\le k_*$. Testing prefix $k$ gives
\[
 K_\alpha(v)\ge A_k+\alpha F_k(v-s_k)
 \ge A_k+\alpha f_k(v-s_k)
 =(1+\alpha)A_k-\alpha A_{k-1}.
\]
Thus $A_k\le(1-q)K_\alpha(v)+qA_{k-1}$. Starting with $A_0=0$, induction yields
$H(v)=A_{k_*}\le(1-q^{k_*})K_\alpha(v)\le(1-q^m)K_\alpha(v)$.
This also covers $\alpha=0$.

For \cref{hs:eq:cutoff-slope}, use the prefix through $s_k\le x$; if there is none,
$F(x)=0$. For a nonempty prefix the tested objective equals
\[
 H(x)+(1+\alpha)(y-x)F(x)+\alpha(x-s_k)F(x).
\]
The omitted terms are nonnegative. This proves the claim, including atoms at $x$.
\end{proof}

Each additional seller type contributes one step to the recurrence, which
explains the factor $q^m$. With no support-size bound we still obtain the
useful inequality $H(v)\le K_\alpha(v)$.

We now add the buyer's incentive constraint. For two buyer types, only the
high type can gain by reporting a lower value, so one information-rent term
suffices. The second cutoff inequality controls exactly this term.
\begin{lemma}[Binary-Buyer Bilateral Bound]\label{hs:lem:bilateral}
A buyer with at most two types and an independent seller with at most $m$ types
satisfy $V_\alpha\ge D_m(\alpha/(1+\alpha))^{-1}\FB$ for every $\alpha\ge0$.
\end{lemma}
\begin{proof}[Proof of \cref{hs:lem:bilateral}]
A buyer with only one possible value $y$ has $V_\alpha\ge K_\alpha(y)$, and
\cref{hs:eq:cutoff-bound} suffices because $D_m(q)\ge1-q^m$.
Now write the buyer's values as $x<y$, with high probability $t\in(0,1)$. Set
\begin{equation}\label{hs:eq:z}
 z=x-q\frac{t}{1-t}(y-x).
\end{equation}
The combined buyer scores $b+\alpha\phi(b)$ are $(1+\alpha)z$ and $(1+\alpha)y$.
For each of these coefficients choose a maximizing seller prefix in
\cref{hs:eq:K}, breaking ties by the largest prefix. The chosen prefix is
nondecreasing in its argument: adding the two maximizing inequalities at
$v<v'$ gives $(v'-v)(F_{k'}-F_k)\ge0$. The two cutoff allocations are therefore
pointwise monotone in both reports and are implementable. Their Lagrangian is
\begin{equation}\label{hs:eq:two-cutoffs}
 V_\alpha\ge (1-t)K_\alpha(z)+tK_\alpha(y).
\end{equation}
The Lagrangian permits expected revenue of either sign, so the two cutoff
rules need not balance the budget separately. Budget balance will enter when
we minimize over $\alpha$.

Convexity of $H$, or a direct pointwise calculation of positive parts, gives
\[
 H(x)-H(z)\le(x-z)F(x).
\]
Since $(1-t)(x-z)=qt(y-x)$, \cref{hs:lem:cutoff} implies
\begin{align}
 (1-t)H(x)
 &\le(1-q^m)(1-t)K_\alpha(z)+qt(y-x)F(x)\notag\\
 &\le(1-q^m)(1-t)K_\alpha(z)+q(1-q)tK_\alpha(y).
 \label{hs:eq:binary-rent-bound}
\end{align}
Applying \cref{hs:eq:cutoff-bound} again at $y$ gives
\begin{align*}
 \FB&=(1-t)H(x)+tH(y)\\
 &\le(1-q^m)(1-t)K_\alpha(z)
       +\bigl(1-q^m+q(1-q)\bigr)tK_\alpha(y)\\
 &\le D_m(q)\bigl((1-t)K_\alpha(z)+tK_\alpha(y)\bigr)
 \le D_m(q)V_\alpha.
\end{align*}
Both cutoff values are nonnegative. No support interleaving, density, or
regularity condition was used.
\end{proof}

The bilateral bound holds for every choice of probabilities on the two
buyer values. It therefore has exactly the uniformity required by the
finite-support construction in \cref{thm:compatible-targets}.

\begin{proof}[Matching lower bound in \cref{thm:curve}]
Fix $\alpha$ and apply the seller regularization of
\cref{sec:global-ironing-v2} to each full seller distribution. This preserves
seller support sizes and leaves buyer supports unchanged. Every probability
distribution on a binary buyer support remains binary, and every seller prior conditioned below a cap
has at most $m$ types. Thus \cref{hs:lem:bilateral} supplies all
bilateral inequalities required by \cref{thm:compatible-targets}, with factor
$D_m(q)^{-1}$. Composition via \cref{eq:cap-target-bound}, followed by
\cref{eq:regularization-summary}, gives the same factor for the original
matching market.
\end{proof}

This application uses the reduction for fixed buyer supports; it does not
require a bound for buyers with arbitrarily many values.

\subsection{Tight Instances and the Resulting Constants}\label{app:tight-family}
To show that the guarantee cannot be improved, it suffices to construct
bilateral instances approaching it. We choose geometrically spaced
cumulative seller probabilities so that only one buyer--seller type pair
has positive Lagrangian weight. This makes the optimum explicit, while
approaching equality in both the seller recurrence and the buyer's
information-rent bound.

Fix $m\ge1$ and $\alpha>0$. Choose small positive rational parameters
$\eta,t<1$. Set
\begin{equation}\label{eq:family}
 \begin{gathered}
 F_k=\eta^{m-1-k}\quad(0\le k<m),\qquad
 f_0=F_0,\quad f_k=F_k-F_{k-1}\ (k\ge1),\\
 r_k=\left(\frac{\alpha\eta}{1+\alpha-\eta}\right)^k,
 \qquad s_k=1-r_k,\\
 \Pr[B=x]=1-t,\quad \Pr[B=1]=t,
 \qquad x=\frac{\alpha t}{1+\alpha-t}.
 \end{gathered}
\end{equation}
Choose $t$ small enough that $0<x<s_1$ when $m\ge2$; when $m=1$, the sole
seller cost is zero. The seller costs are strictly increasing and below one.
All input values already lie in $[0,1]$. In this construction, the low
buyer value cancels its information-rent term, and each positive seller cost
cancels the high buyer's Lagrangian coefficient. The next calculation
verifies these cancellations and computes the resulting ratio.

\begin{lemma}[A Tight Bilateral Family]\label{lem:family}
The low buyer's combined score $b+\alpha\phi(b)$ is zero. The lowest
seller's combined score $s+\alpha\tau(s)$ is zero, and every other seller
type has combined score $1+\alpha$. Thus
\begin{equation}\label{eq:family-V}
 V_\alpha=(1+\alpha)t f_0.
\end{equation}
Moreover,
\begin{equation}\label{eq:family-F}
 \frac{\FB}{tf_0}
 =\frac{\alpha(1-t)}{1+\alpha-t}
   +1+\alpha-\alpha
       \left(\frac{\alpha}{1+\alpha-\eta}\right)^{m-1}.
\end{equation}
As $\eta,t\downarrow0$, its ratio $V_\alpha/\FB$ converges to $D_m(q)^{-1}$.
\end{lemma}
\begin{proof}
The low buyer coefficient is
$[(1+\alpha-t)x-\alpha t]/(1-t)=0$.
For $k\ge1$ we have $F_{k-1}/F_k=\eta$ and
\[
 r_k=\frac{\alpha F_{k-1}}{(1+\alpha)F_k-F_{k-1}}r_{k-1}.
\]
Substitution in the discrete virtual-cost formula gives
$(1+\alpha)s_k+\alpha F_{k-1}(s_k-s_{k-1})/f_k=1+\alpha$.
Hence only the type pair $(B,S)=(1,0)$ has strictly positive Lagrangian
weight. Trading only at this pair is monotone in both reports, so it attains
the upper bound obtained by keeping every positive weight. This proves
\cref{eq:family-V}.

Summing the seller coefficients through the final type gives
\[
 \E[1-S]+\alpha(1-s_{m-1})=(1+\alpha)f_0.
\]
Also $(1-s_{m-1})/f_0=(\alpha/(1+\alpha-\eta))^{m-1}$.
The low buyer trades efficiently only with the lowest seller, so
$\FB=(1-t)f_0x+t\E[1-S]$, proving \cref{eq:family-F}.
Its limit is $(1+\alpha)D_m(q)$. This proves the claimed limiting ratio.
\end{proof}

We have now matched the Lagrangian lower bound for each fixed multiplier.
It remains to minimize over the multiplier to obtain the budget-balanced
ratios and evaluate the constants.

\begin{proof}[Completion of \cref{thm:curve}]
The family proves the matching upper bound because a single edge is allowed.
At $\alpha=0$, $V_0=\FB$. For irrational fixed $\alpha$, approximate the finitely
many support values in each fixed member of the family by rationals;
finite-dimensional objective continuity preserves the limiting ratio.

For $m\ge2$, $D_m'(q)=1-2q-mq^{m-1}$ and
$D_m''(q)=-2-m(m-1)q^{m-2}<0$. The derivative is positive at zero and negative
at $1/2$, so the unique maximizer $q_m\in(0,1/2)$ satisfies
\begin{equation}\label{eq:root}
 1-2q_m-mq_m^{m-1}=0.
\end{equation}
Set $q_1=0$. The case $m=1$ has
$D_1(q)=1-q^2$ and maximum one at zero. Budget duality gives the corresponding
lower bounds for $\SB$. Applying weak duality to the tight family at a
minimizing multiplier gives the reverse worst-case inequalities; no exchange
of a maximum with an infimum is required.

For $0<q<1$, $D_{m+1}(q)>D_m(q)$, so $\gamma_m$ decreases strictly. Also
\[
  \frac{5}{4}-2^{-m}=D_m(1/2)\le\max_qD_m(q)\le\frac{5}{4}.
\]
Thus $\gamma_m\downarrow4/5$. Pointwise in every finite $\alpha$, the infimum
over $m$ gives
\begin{equation}\label{eq:one-sided-curve}
 \beta_\infty(\alpha)=\frac{(1+\alpha)^2}{1+3\alpha+\alpha^2},
\end{equation}
whose minimum is $4/5$ at $\alpha=1$. For $m=2,3$, the maximizing $q$ is
$1/4,1/3$, giving $8/9,27/32$, respectively.
\end{proof}

\paragraph{Exact Optima for Finite Instances.}
The same construction also explains the gap between a finite numerical
example and the worst-case limit. Whenever an instance in \cref{eq:family}
has $V_\alpha<\FB$, its chosen $\alpha$ is already an optimal budget
multiplier. To see this, let $g_*=tf_0$ be the gains and expected revenue
from trading only at $(B,S)=(1,0)$, and write
$L=V_\alpha=(1+\alpha)g_*$. Every other type pair with positive surplus has
zero Lagrangian weight. The efficient rule therefore also has Lagrangian
value $L$, with expected revenue
\[
 -d=\frac{L-\FB}{\alpha}<0.
\]
Mixing these two rules gives every expected revenue in $[-d,g_*]$ while
keeping Lagrangian value $L$. In particular, the mixture with zero expected
revenue attains $\SB=L$ for each such finite instance.

For example, Schottm\"uller's reported binary grid minimum~\cite{Schott23}
uses buyer values $1/100,1$ with high-type probability $1/25$, and seller
costs $0,99/100$ with lowest-cost probability $1/25$. This is precisely
\cref{eq:family} with $m=2$, $\eta=t=1/25$, and $\alpha=8/25$. The formulas
above give $\FB=37/15625$ and $\SB=V_\alpha=33/15625$, hence the exact ratio
$33/37=0.891891\ldots$. This agrees with the reported $0.89189$; the smaller
worst-case ratio $8/9$ is approached as the parameters tend to zero.

\section{  {Conclusion} }\label{sec:further-consequences}\label{sec:conclusion}
The   reduction transfers bilateral  second-best bounds   without
loss, while the  binary-support   and MHR applications combine  the same
cap-monotone composition step  with different local bilateral arguments.
The remaining refinements and   the offering-mechanism application  are
independent consequences   rather than  ingredients of the two sharp bounds.
We summarize them before turning to open questions.

  {Keeping}  the seller probabilities in the cutoff recurrence gives sharp
bounds in terms of the probability of the lowest seller cost, both with a
seller support bound (\cref{thm:mass-finite}) and without one
(\cref{thm:bottom-mass}). Only that atom is restricted; the buyer's high type
may still be rare. When the bilateral bound varies across seller caps,
\cref{ct:cor:envelope} constructs cap-dependent targets and
\cref{ct:thm:global} combines them into a matching-market guarantee.
The local linear program in \cref{cor:local-coefficient} computes the best factor that holds
uniformly over the given buyer support and seller caps. These factors give
market-level guarantees through composition; the program does not compute
an entire market's second best.

Public additive match-specific values preserve the single-parameter payment
identities. In particular, the binary-support bounds require only that each
compatible edge have a binary endpoint, not that all binary agents lie on
the same side (\cref{thm:edgewise-premia}). The sharp constants remain those
of the corresponding bilateral classes.

 For specified offering mechanisms,  if $G_{\rm offer}$ is the expected GFT of
the uniform mixture of the generalized seller- and buyer-offering mechanisms,
\cref{lem:offering-comparison} proves $2G_{\rm offer}\ge V_1$.
Together with the finite-support   bounds  this gives the factors in
\cref{cor:offering-guarantees}. These are guarantees for those mechanisms;
the reduction alone does not preserve a supplied bilateral mechanism, and
we make no sharpness claim for these offering factors.

These results suggest three directions. First, what are the sharp factors when both sides
have more than two types? New bilateral bounds would immediately yield
matching-market bounds through the finite-support reduction. Second, can
bilateral value inequalities alone imply cap-monotone rules for
continuous distributions? This would require a continuous counterpart of
the finite minimax and smallest-mass construction, together with measurable
choices and valid limiting arguments. Resolving it would extend the
automatic construction to settings where we currently build the bilateral
rules directly. Third, which other matching-market mechanisms satisfy a
comparison $\lambda G\ge V_\alpha$? Searching systematically for such
inequalities would let us combine a growing collection of bilateral bounds
with guarantees for additional mechanisms, as generalized offering already
illustrates.

\paragraph{AI Disclosure.}
ChatGPT assisted in organizing the manuscript, developing and checking proof
arguments, and writing auxiliary validation programs. The human authors are responsible for
all mathematical claims, attribution, and the final submission.
\printbibliography[heading=bibintoc,title={References}]
\clearpage
\appendix
\crefalias{section}{appendix}
\crefalias{subsection}{appendix}
\crefalias{subsubsection}{appendix}
\section{Payment Identities, Budget Duality, and Regularization}\label{app:foundations}
The reduction uses two properties of the finite mechanism-design problem:
normalized payments express revenue as a linear function of the allocation,
and ironing permits exact optimization of the resulting Lagrangian.
We prove these properties here, including the budget-duality statement used
to convert bounds for every multiplier into a second-best guarantee. We
then give the seller regularization proof deferred from \cref{sec:global-ironing-v2}.

\subsection{Finite Payment Identities and Budget Duality}\label{sec:finite-calculus}
We use the supports and coefficients of \cref{sec:model}. First we derive
the maximum expected revenue compatible with a given allocation. The
single expected-budget constraint will then let us implement a second-best
optimum by mixing at most two Lagrangian maximizers.

\begin{proof}[Proof of \cref{lem:allocation-characterization}]
Put $\Delta b^\ell=b^{\ell+1}-b^\ell$. Fix a buyer and abbreviate $Q_\ell=Q(b^\ell)$.  The adjacent BIC
inequalities are equivalent to $Q_\ell\le Q_{\ell+1}$ and
\[
 \Delta b^\ell Q_\ell
 \le U_{\ell+1}-U_\ell
 \le \Delta b^\ell Q_{\ell+1},
\]
where $U_\ell=b^\ell Q_\ell-T_\ell$ is interim utility and $T_\ell$ is
interim net payment to the mechanism.  Hence every IIR implementation has
at least the utilities obtained from $U_1=0$ and the lower adjacent
envelope $U_{\ell+1}-U_\ell=\Delta b^\ell Q_\ell$.  The corresponding
revenue-maximizing interim payment is
\begin{equation}\label{eq:buyer-envelope-payment}
 T_\ell^{B}=b^\ell Q_\ell-
        \sum_{r<\ell}\Delta b^r Q_r.
\end{equation}
Monotonicity makes these adjacent inequalities telescope to every pair
of reports, so the payments are BIC and leave nonnegative truthful utility.

For a seller put $Y_k=Y(s^k)$ and
$U_k=-T_k-s^kY_k$.  Adjacent BIC is equivalent to
$Y_k\ge Y_{k+1}$ and
\[
 (s^{k+1}-s^k)Y_{k+1}
 \le U_k-U_{k+1}
 \le (s^{k+1}-s^k)Y_k.
\]
The net-payment-maximizing IIR choice has $U_m=0$ and uses the lower
adjacent envelope.  It is
\begin{equation}\label{eq:seller-envelope-payment}
 T_k^{S}=-s^kY_k-
   \sum_{r=k}^{m-1}(s^{r+1}-s^r)Y_{r+1}.
\end{equation}
Again, the adjacent inequalities telescope to all deviations. The resulting
utilities are nonnegative, and every other BIC, IIR payment rule leaves at
least these utilities at each type.

Taking expectations in \cref{eq:buyer-envelope-payment} and exchanging
the order of summation gives $\E[T^B]=\E[\phi(v)Q(v)]$.
Likewise, \cref{eq:seller-envelope-payment} gives
$\E[T^S]=-\E[\tau(c)Y(c)]$.  Summing over agents shows that
$\Lambda(x)$ is the largest expected net revenue compatible with BIC and
IIR for the allocation $x$.  Therefore ex-ante WBB is possible only if
$\Lambda(x)\ge0$, and the displayed normalized payments prove sufficiency.
\smallskip
\noindent\emph{Budget duality.}
The payment identities reduce budget feasibility to the single linear
inequality $\Lambda(x)\ge0$. We now prove the duality and two-allocation
implementation claims of the lemma.
The polytope $\cX$ has finitely many vertices. Write their gain--budget pairs
as $(G_r,L_r)$. Then $V_\alpha=\max_r(G_r+\alpha L_r)$ is a continuous
convex piecewise-affine function. The no-trade allocation ensures that its
largest eventual slope is nonnegative. If that slope is positive,
$V_\alpha\to\infty$; if it is zero, the function is eventually constant.
Thus a minimum is attained at a finite $\alpha^*\ge0$.

For a fixed multiplier, call a vertex active when its affine function
attains $V_\alpha$. The right derivative is the largest active slope,
and the left derivative is the smallest: all inactive lines have a
strict gap and remain inactive sufficiently near the given multiplier.
If $\alpha^*=0$, the right derivative is nonnegative, so an active vertex
has $L_r\ge0$ and $G_r=V_0$. If $\alpha^*>0$, the active slopes straddle
zero. An active zero-slope vertex is already feasible and has gains
$V_{\alpha^*}$. Otherwise select active slopes $L_-<0<L_+$ and mix their
vertices with weights $L_+/(L_+-L_-)$ and $-L_-/(L_+-L_-)$.
The mixture has budget zero and gains $V_{\alpha^*}$, since both active
lines equal that value at the same multiplier.

Finally, every $x\in\cX$ with $\Lambda(x)\ge0$ satisfies
$\GFT(x)\le\mathcal L_\alpha(x)\le V_\alpha$ for every $\alpha\ge0$.
This weak-duality inequality and the constructed feasible optimizer prove
both the identity and the implementation alternatives.
\end{proof}

\paragraph{Strong Budget Balance.}\label{app:sbb-conversion}
For independent finite types and unrestricted signed transfers, ex-ante WBB
and profilewise SBB give the same allocation benchmark. To see this, take
interim net payments $T_a$, with means $\mu_a$, and let $R=\sum_a\mu_a\ge0$.
For $n\ge2$ agents use
\[
 p_a(\theta)=T_a(\theta_a)-\frac{1}{n-1}\sum_{h\ne a}
                 \left(T_h(\theta_h)-\mu_h+\frac{R}n\right).
\]
Their sum is identically zero. Every report by agent $a$ has expected payment
$T_a(\theta_a)-R/n$, so BIC is unchanged and IIR weakly improves. Fewer than
two agents permit no trade. This conversion does not assert universal DSIC
or ex-post IR.

\subsection{Fixed-Multiplier Optimization}\label{app:optimizer}
\label{sec:transfer-proofs}

Having reduced the budget constraint to a multiplier, we now prove that
the ironed matching rule optimizes the resulting objective. The proof also
checks that every selected edge creates positive actual surplus.

For a fixed multiplier, abbreviate
$a_{i\ell}=a_i^\alpha(b_i^\ell)$ and $d_{jk}=d_j^\alpha(s_j^k)$.
Bars denote the corresponding ironed coefficients. Each block produced by   weighted ironing  has coefficient equal
to its probability-weighted mean.

We first bound these block coefficients by the actual values and costs.

\begin{lemma}
  \label{lem:ironing-endpoint-bounds}
 For every $\alpha\ge0$, every intermediate buyer block $I$ and seller block
$J$   produced by weighted ironing  satisfy
\[
 a_I\le(1+\alpha)\min_{\ell\in I}b_i^\ell,
 \qquad
 d_J\ge(1+\alpha)\max_{k\in J}s_j^k.
\]
In particular, $\bar a_i(b;\alpha)\le(1+\alpha)b$ and
$\bar d_j(s;\alpha)\ge(1+\alpha)s$ at every report.
\end{lemma}

\begin{proof}
The singleton bounds follow from $\phi_i(b)\le b$ and $\tau_j(s)\ge s$.
When adjacent blocks with means $u\ge v$ are pooled, their positive-weight
mean $w$ lies in $[v,u]$.  For a buyer, inherit the upper bound from the
left block: its smallest value is also the smallest value of the union and
$w\le u$.  For a seller, inherit the lower bound from the right block:
its largest cost is the largest cost of the union and $w\ge v$.
Induction proves both invariants, including equal-mean merges and $\alpha=0$.
\end{proof}

The block bounds will ensure positive surplus. To prove optimality, we
first show that ironing can only increase the objective of a monotone
allocation, then verify that the canonical matching rule attains equality.

\begin{proof}[Proof of \cref{lem:exact-ironed-optimizer}]
Fix any $x\in\cX$.  For buyer $i$, let $Q_{i\ell}$ be the interim service
probability at $b_i^\ell$.  It is nondecreasing in $\ell$.    {Consider one ironing
 merge of adjacent blocks }  $A,B$ having masses $p_A,p_B$, coefficient means
$a_A\ge a_B$, and weighted-average services $Q_A\le Q_B$.  Replacing the two
coefficient means by their common weighted average changes the buyer contribution
by
$\frac{p_Ap_B}{p_A+p_B}(a_A-a_B)(Q_B-Q_A)\ge0$.

Iterating over all merges gives
  \begin{equation}
 \sum_\ell g_{i\ell}a_{i\ell}Q_{i\ell}
 \le
 \sum_\ell g_{i\ell}\bar a_{i\ell}Q_{i\ell}.
 \label{eq:buyer-ironing-global}
\end{equation}

For seller $j$, the interim service probabilities $Y_{jk}$ are
nonincreasing.  The same two-block calculation, now remembering the minus
sign on the seller coefficient, gives
  \begin{equation}
 -\sum_k f_{jk}d_{jk}Y_{jk}
 \le
 -\sum_k f_{jk}\bar d_{jk}Y_{jk}.
 \label{eq:seller-ironing-global}
\end{equation}
Consequently,
\begin{align}
 \mathcal L_\alpha(x)
 &\le
 \E\!\left[\sum_{(i,j)\in E}
   (\bar a_i(v_i;\alpha)-\bar d_j(c_j;\alpha))x_{ij}\right]\notag\\
 &\le
 \E\!\left[\max_{M\in\cF}
   \sum_{(i,j)\in M}
   (\bar a_i(v_i;\alpha)-\bar d_j(c_j;\alpha))\right].
 \label{eq:ironed-oracle-upper}
\end{align}
The oracle rule attains the second inequality profilewise.

We next verify implementability.  Fix all reports except buyer $i$'s and write
$q=\bar a_i(v_i;\alpha)$.  Every feasible matching serving $i$ has weight
$q+C_M$, while every matching not serving $i$ has weight independent of $q$.
Among the served matchings, the maximum intercept, minimum cardinality, and
fixed-order choice are all independent of $q$.  Hence the canonical outcome
switches from not serving $i$ to serving $i$ at most once as $q$ rises.
Because $q$ is nondecreasing in $v_i$, buyer service is nondecreasing.
For a seller the same argument has parameter $q=\bar d_j(c_j;\alpha)$ and
slope $-1$, so service is nonincreasing.  Thus $X^\alpha\in\cX$.

Finally, all reports in one buyer ironing block induce the same incident
edge-weight vector.  Since the oracle's tie rule is report independent, its
outcome, and hence that buyer's interim allocation, is constant on the block.
The analogous statement holds for every seller block.  Therefore equality
holds in both   \cref{eq:buyer-ironing-global} and
\cref{eq:seller-ironing-global}  for $X^\alpha$.  It attains the upper bound in
\cref{eq:ironed-oracle-upper}, proving optimality.

If a selected edge had negative ironed weight, deleting it would improve
the objective; if its weight were zero, deletion would preserve the objective
and reduce cardinality.  Both contradict the oracle rule and downward
closure.  Thus every selected $(i,j)$ has
$0<\bar a_i(b;\alpha)-\bar d_j(s;\alpha) \le(1+\alpha)(b-s),$

by   \cref{lem:ironing-endpoint-bounds}, proving $b>s$.

Weighted   ironing  uses linearly many merges.   Evaluation of the resulting
rule  then uses one matching call.
\end{proof}

\subsection{Seller Regularization}\label{app:regularization-proof}
We prove the regularization statement used in \cref{sec:global-ironing-v2}.
The construction recovers actual seller costs from ironed scores; the
four lemmas below establish that these costs form a valid prior, improve
the first-best benchmark, and preserve the Lagrangian optimum.

Fix one seller with costs $s_1<\cdots<s_m$ and masses $f_k$.  Put
\[
 F_k=\sum_{h\le k}f_h,
 \qquad C_k=\sum_{h\le k}f_hs_h,
 \qquad \Phi_k=C_k+\alpha F_ks_k,
\]
with $F_0=C_0=\Phi_0=0$ and $d_k^\alpha=s_k+\alpha\tau(s_k)$. By direct subtraction,
\begin{equation}\label{eq:phi-slopes-v2}
 \frac{\Phi_k-\Phi_{k-1}}{f_k}=d_k^\alpha.
\end{equation}
Join the points $(F_k,\Phi_k)$ linearly, and let
$\overline\Phi$ be their greatest convex minorant, the pointwise largest
convex function below this polygonal curve. It touches both endpoints.
Define its slopes
$\overline d_k= \frac{\overline\Phi(F_k)-\overline\Phi(F_{k-1})}{f_k},$
which are nondecreasing. These slopes are exactly the   ironed  scores:
replacing a decreasing block of slopes by its weighted mean replaces that
part of the curve by a chord, and the completed merges give this convex minorant.

Keep the probabilities $f_k$, but define reconstructed costs
$\widehat s_k$ recursively by
\begin{equation}\label{eq:reconstructed-cost-recursion-v2}
 \widehat C_k+\alpha F_k\widehat s_k
 =\overline\Phi(F_k),
 \qquad
 \widehat C_k=\sum_{h\le k}f_h\widehat s_h.
\end{equation}
The coefficient of $\widehat s_k$ is
$f_k+\alpha F_k>0$, so the recursion is well defined.

The recursion must define a genuine ordered prior, rather than merely a list of convenient coefficients. Its consecutive differences verify this.

\begin{lemma}
\label{lem:pseudo-valid-v2}
For $\alpha>0$, the sequence $(\widehat s_k)$ is strictly increasing,
and its combined seller scores are exactly $(\overline d_k)$.
Hence the ironed seller distribution is $\alpha$-weakly regular.  For
$\alpha=0$ no transformation is needed: $d_k^0=s_k$ is already
increasing.
\end{lemma}

\begin{proof}
Subtract consecutive equations in
\cref{eq:reconstructed-cost-recursion-v2}:
\begin{equation}\label{eq:pseudo-slope-v2}
 \overline d_k=(1+\alpha)\widehat s_k
 +\alpha\frac{F_{k-1}}{f_k}
       (\widehat s_k-\widehat s_{k-1}),
\end{equation}
with the final term absent at $k=1$.  These are exactly the discrete
combined seller scores of the ironed seller distribution.

It remains to prove strict increase.  The first original slope is
$d_1^\alpha=(1+\alpha)s_1$, while
$d_h^\alpha\ge(1+\alpha)s_h>d_1^\alpha$ for every $h>1$.
Consequently every chord from the origin to a later breakpoint has
slope strictly larger than $d_1^\alpha$; the greatest convex minorant
therefore touches the first breakpoint and
$\overline d_1=d_1^\alpha$.  Hence $\widehat s_1=s_1$.
For $k=2$, every block made from later original slopes has mean strictly above
$d_1^\alpha$; hence $\overline d_2>\overline d_1=(1+\alpha)\widehat s_1$.
For $k>2$, if $\widehat s_{k-1}>\widehat s_{k-2}$, then
\cref{eq:pseudo-slope-v2} gives
$\overline d_{k-1}>(1+\alpha)\widehat s_{k-1}$; monotonicity of the hull
slopes yields the same strict inequality with $\overline d_k$ in place
of $\overline d_{k-1}$.  Solving \cref{eq:pseudo-slope-v2} for the
increment gives
\[
 \widehat s_k-\widehat s_{k-1}
 =\frac{f_k\bigl(\overline d_k-(1+\alpha)
        \widehat s_{k-1}\bigr)}{f_k(1+\alpha)+\alpha F_{k-1}}>0.
\]
Induction completes the proof.
\end{proof}

Because both $(s_k)$ and $(\widehat s_k)$ are strictly increasing in the
same index, the original and reconstructed-cost markets have the same indexed
feasibility and interim-monotonicity polytope.  Only the seller coefficients in the
Lagrangian change.

To compare first-best gains, we use cumulative costs: the cost averaged
over every initial probability interval can only fall. Write $s(u)$ and
$\widehat s(u)$ for the quantile functions, equal to $s_k$ and
$\widehat s_k$ on $(F_{k-1},F_k]$, respectively.

\begin{lemma}\label{lem:tail-order-v2}
For every quantile $q\in[0,1]$,
\[
 \int_0^q\widehat s(u)\dd u
 \le \int_0^qs(u)\dd u.
\]
Consequently, for every  $\kappa$,
\begin{equation}\label{eq:put-order-v2}
 \E[(\kappa-\widehat S)_+]
 \ge\E[(\kappa-S)_+].
\end{equation}
\end{lemma}

\begin{proof}
At a breakpoint let
\[
 D_k=C_k-\widehat C_k,
 \qquad
 g_k=\Phi_k-\overline\Phi(F_k)\ge0.
\]
Using \cref{eq:reconstructed-cost-recursion-v2} and
$s_k-\widehat s_k=(D_k-D_{k-1})/f_k$,
\[
 g_k=D_k+\alpha F_k\frac{D_k-D_{k-1}}{f_k},
\]
so
\begin{equation}\label{eq:D-recursion-v2}
 D_k=\frac{f_kg_k+\alpha F_kD_{k-1}}{f_k+\alpha F_k}\ge0.
\end{equation}
Linear interpolation proves the cumulative inequality at every
quantile. For any nondecreasing quantile function $a$,
\[
 \int_0^1(\kappa-a(q))_+\dd q
 =\max_{0\le q\le1}
   \left\{\kappa q-\int_0^qa(u)\dd u\right\}.
\]
The integrand $\kappa-a(u)$ is nonincreasing, so the part where it is
positive is an initial interval. Integrating to the end of that interval
proves the displayed maximum identity; a flat interval at equality adds
zero, and the empty and full intervals cover both endpoint cases.
The cumulative inequality gives \cref{eq:put-order-v2}.
\end{proof}

After other types are fixed, first-best gains are a constant plus a positive-part term $(\kappa-c)_+$ in one seller's cost.
The expectation comparison in \cref{lem:tail-order-v2} therefore applies directly.

\begin{lemma}
\label{lem:fb-rises-v2}
Replacing one seller by its ironed seller distribution weakly increases
first-best gains under every downward-closed family of matchings.
\end{lemma}

\begin{proof}
Fix all other types.  The best matching not using this seller has a
constant value $A$.  If no feasible matching uses the seller, the
replacement has no effect.  Otherwise every matching using the seller
contains it once, so the best such value is $B-c$ for a finite
constant $B$.  The conditional first-best value is
\[
 \max\{A,B-c\}=A+(\kappa-c)_+,
 \qquad \kappa=B-A.
\]
For every realization of the other agents' types, $A$ and $B$ are constants
with respect to this seller's cost, so \cref{eq:put-order-v2} applies
conditionally.  Average over the other agents.  Sellers may then be replaced
one at a time; each replacement changes only one marginal and preserves
independence of the product prior.
\end{proof}

Unlike the benchmark, the optimized Lagrangian is unchanged: the two markets have identical ironed edge weights at each coupled type-index profile.

\begin{lemma}
\label{lem:value-preserved-v2}
Let $F$ be the original profile of priors and $\widehat F^\alpha$ the profile
obtained by reconstructed-cost replacement for every seller.  Then
\[
  V_\alpha(F)=V_\alpha(\widehat F^\alpha).
\]
\end{lemma}

\begin{proof}
Couple the two markets by the common type indices.  They have the same
index probabilities and, at every index profile, the same ironed Lagrangian
coefficients: buyer scores are unchanged, while each transformed seller's
unironed scores are already nondecreasing and equal the original seller's
ironed scores.
With the same canonical tie rule, \cref{lem:exact-ironed-optimizer} therefore
selects the same matching at every coupled profile, for every compatibility
graph and matching-based feasibility family.  The profilewise ironed objective
and hence its expectation are identical.
\end{proof}

\begin{proof}[Proof of \cref{prop:regularization-main}]
\Cref{lem:pseudo-valid-v2} gives an ordered, weakly regular prior with the same
masses, and \cref{lem:value-preserved-v2} preserves the Lagrangian optimum.
Apply \cref{lem:fb-rises-v2} one seller at a time for first-best dominance;
independent marginal replacement preserves the product prior. The   {support
points}  can change, but both the uniform and support-size premises allow
arbitrary finite   {support points}.
\end{proof}

\section{Extensions to General Priors}\label{app:general-priors}
This appendix justifies the extensions beyond finite priors stated in
\cref{sec:model,sec:mhr-overview}. One-sided rounding   handles  bounded Borel
priors, and interval projection then transfers every uniform finite-prior
second-best guarantee to independent Borel priors  with finite
expected first-best gains.   The separate clipping argument extends the MHR
bounds under their stated lower-support condition.
For Borel priors below, $\SB$ denotes the supremum of expected GFT over
BIC, IIR, ex-ante-WBB mechanisms with integrable transfers and truthful
utilities.

\subsection{  {General Borel  Priors and the Half Guarantee } }\label{sec:borel}

  We first  transfer a finite guarantee to   priors supported on a bounded
interval by rounding  buyer values down and seller costs up, and then remove
the support bounds by projecting reports onto expanding intervals. Let  $R$
be the maximum cardinality of a feasible matching.  For the bounded step,
suppose all types lie in a common interval $[L,U]$.

For an allocation with bounded monotone interim services on $[L,U]$,
define the following net payments, which depend only on the agent's own
report:
\begin{equation}\label{eq:borel-envelope}
 T_i(v)=vQ_i(v)-\int_L^v Q_i(t)\,\dd t,\qquad
 T_j(s)=-sY_j(s)-\int_s^U Y_j(t)\,\dd t.
\end{equation}
For any buyer report $r$, truthful utility minus deviation utility is
$\int_r^v(Q_i(t)-Q_i(r))\,\dd t\ge0$.
For a seller it is
$\int_s^r(Y_j(t)-Y_j(r))\,\dd t\ge0$.
Both inequalities follow from monotonicity, with the usual orientation
of an integral when its limits are reversed. Truthful utilities are the
nonnegative displayed envelope integrals. These formulas therefore give
BIC and IIR directly, including at atoms; no differentiability of a prior
is required. Write
$\Lambda_{\rm gen}(x)=\sum_i\E[T_i(v_i)]+\sum_j\E[T_j(c_j)]$
for the expected net revenue of this implementation.

The next lemma extends a rounded allocation and its normalized payments
together. Defining the allocation also at zero-probability reports is
necessary because incentive compatibility compares truthful reporting with
every allowed deviation.

\begin{lemma}\label{lem:grid-extension-v2}
Fix a grid containing $L$ and $U$, with consecutive gaps at most
$\delta$.  Round buyers down to lower grid endpoints and sellers up to upper
grid endpoints.  Any feasible interim-monotone allocation for the occupied
rounded reports extends to the full report domains with the same generalized
expected revenue.  Its actual GFT is at least its rounded GFT.  Revenue
preservation holds for either sign of the discrete revenue functional.
\end{lemma}

\begin{proof}
Let $\widehat{\mathcal T}_i^B$ and $\widehat{\mathcal T}_j^S$ be the occupied
supports of the rounded buyer and seller marginals.  For a buyer report $v$,
let
\[
 \rho_i(v)=\max\{b\in\widehat{\mathcal T}_i^B:
                    b\le \lfloor v\rfloor_\delta\},
\]
when the set is nonempty, and write $\rho_i(v)=\bot$ otherwise.  For a seller
report $c$, define
\[
 \sigma_j(c)=\min\{s\in\widehat{\mathcal T}_j^S:
                     s\ge \lceil c\rceil_\delta\},
\]
with $\sigma_j(c)=\bot$ if the set is empty.  Both maps are nondecreasing where finite. Composing them with
nondecreasing buyer service or nonincreasing seller service preserves the
respective incentive direction.

To define the allocation at an arbitrary report profile, replace every active
report by its image under $\rho$ or $\sigma$, replace an inactive buyer report by its smallest occupied type and an
inactive seller report by its largest occupied type, evaluate the finite rule, and
then delete every edge incident to an inactive agent.  The result is feasible
by downward closedness.  Under the truthful product distribution, every
opponent is active almost surely and its image has exactly the corresponding
rounded marginal.  Hence the extended interim services are
\[
 Q_i^{\rm ext}(v)=
 \begin{cases}Q_i^{\rm grid}(\rho_i(v)),&\rho_i(v)\ne\bot,\\0,&\rho_i(v)=\bot,
 \end{cases}
 \qquad
 Y_j^{\rm ext}(c)=
 \begin{cases}Y_j^{\rm grid}(\sigma_j(c)),&\sigma_j(c)\ne\bot,\\0,&\sigma_j(c)=\bot.
 \end{cases}
\]
They are monotone, and at every occupied rounded report they equal the finite
interim services.  Thus the extension is interim monotone; no pointwise
monotonicity at zero-probability reports is required.

The buyer service is constant on each interval mapped to an occupied report
$a$.  For every $v$ in such an interval,
\[
 vQ_a-\int_L^vQ(t)\dd t
 =aQ_a-\int_L^aQ(t)\dd t.
\]
Similarly, if a seller interval maps to the occupied report $u$, then
\[
 -cY_u-\int_c^UY(t)\dd t
 =-uY_u-\int_u^UY(t)\dd t.
\]
The integral-envelope payments are therefore constant on the sets of reports mapped to each value of
$\rho_i$ and $\sigma_j$. Since the images of truthful reports have exactly
the rounded marginal distributions, expected buyer and seller payments equal the normalized
discrete envelope payments exactly.  This proves preservation of
$\Lambda_{\rm gen}$, irrespective of its sign, and the envelopes implement the
extension by a BIC, IIR mechanism.  A nonnegative discrete revenue remains ex-ante WBB.

Finally, every allocated edge has active endpoints and satisfies
\[
 v_i-c_j\ge \rho_i(v_i)-\sigma_j(c_j).
\]
Averaging gives actual GFT at least the rounded GFT.
\end{proof}

\paragraph{  From finite priors to arbitrary Borel priors.}\label{app:borel-half}
  For the unbounded extension, the only additional ingredient needed from the
finite problem is that second best can be attained without trading on
nonpositive-surplus edges.

  \begin{lemma}[Positive-Surplus Finite Optimum]
\label{lem:positive-surplus-finite-optimum}
For independent finite priors and  downward-closed   matching feasibility,
$\SB$ is attained by a feasible interim-monotone allocation that trades only
when $v_i>c_j$ and has nonnegative normalized expected revenue.
\end{lemma}
\begin{proof}
Intersect $\mathcal X$  with the   linear constraints
$x_{ij}(v,c)=0$ whenever $v_i\le c_j$. By
\cref{lem:exact-ironed-optimizer}, for every $\alpha\ge0$ the Lagrangian
maximum $V_\alpha$ has an optimizer in this restricted polytope. Hence the
restricted and unrestricted Lagrangian values coincide for every multiplier.
Applying the same finite linear-programming duality as in
\cref{lem:allocation-characterization} to the restricted polytope shows that
its budget-feasible optimum is
$\min_{\alpha\ge0}V_\alpha=\SB$. An optimizer therefore has the stated
properties.
\end{proof}

Fix $U>0$ and define the interval projection
\[
 \Pi_U(t)=\max\{-U,\min\{t,U\}\},\qquad
 v_i^U=\Pi_U(v_i),\quad c_j^U=\Pi_U(c_j).
\]
Let $\FB_U$ be first best in the projected market on $[-U,U]$.

\begin{lemma}[Interval-Projection Extension]
\label{lem:interval-projection-extension}
Consider a BIC, IIR mechanism on the full report interval $[-U,U]$ for the
projected market. Suppose it trades only on positive-surplus edges and uses
the normalized envelope transfers in \cref{eq:borel-envelope}. Composing every
report with $\Pi_U$ gives a feasible BIC, IIR mechanism for the original market with
the same payment distribution and weakly larger GFT. Expected revenue is
preserved, so ex-ante WBB is preserved whenever it holds in the projected
market.
\end{lemma}
 \begin{proof}
  Projection is measurable and does not affect feasibility. Positive-surplus
trade implies $Q_i(-U)=0$ and $Y_j(U)=0$, so the normalized boundary payments
are $T_i(-U)=T_j(U)=0$.

Fix a buyer's true value $v_i$ and write $v_i^U=\Pi_U(v_i)$. For any
alternative projected report $b_i\in[-U,U]$,
\begin{align*}
 &v_iQ_i(v_i^U)-T_i(v_i^U)-v_iQ_i(b_i)+T_i(b_i)\\
 &=\bigl[v_i^UQ_i(v_i^U)-T_i(v_i^U)-v_i^UQ_i(b_i)+T_i(b_i)\bigr]
 +(v_i-v_i^U)\bigl[Q_i(v_i^U)-Q_i(b_i)\bigr]\ge0.
\end{align*}
The bracket is nonnegative by BIC on $[-U,U]$. The last term vanishes inside
the interval; it is nonnegative above $U$ by buyer monotonicity and below
$-U$ because $Q_i(-U)=0$.

For a seller's true cost $c_j$, write $c_j^U=\Pi_U(c_j)$. For any
alternative projected report $s_j\in[-U,U]$,
\begin{align*}
 &-T_j(c_j^U)-c_jY_j(c_j^U)+T_j(s_j)+c_jY_j(s_j)\\
 &=\bigl[-T_j(c_j^U)-c_j^UY_j(c_j^U)+T_j(s_j)+c_j^UY_j(s_j)\bigr]
 +(c_j^U-c_j)\bigl[Y_j(c_j^U)-Y_j(s_j)\bigr]\ge0.
\end{align*}
The bracket is nonnegative by BIC. The last term is nonnegative below
$-U$ by seller monotonicity and above $U$ because $Y_j(U)=0$. Every
deviation is evaluated through a projected report, proving BIC on the full
report space. Buyers below $-U$ and sellers above $U$ receive zero service and
zero payment; outside the other two boundaries, monotonicity adds a
nonnegative utility increment to the boundary-type utility. Thus IIR is also
preserved.

Under truthful reporting, projection induces exactly the projected product
prior, so the payment distribution is unchanged. If an edge trades, then
$v_i^U>c_j^U$, which implies $v_i\ge v_i^U$ and $c_j\le c_j^U$. Hence
$v_i-c_j\ge v_i^U-c_j^U>0$, proving the GFT comparison.
\end{proof}

\begin{lemma}[Integrability of the Projection Extension]
\label{lem:interval-projection-integrability}
If the original independent Borel priors satisfy $\FB<\infty$,
the transfers and truthful utilities in
\cref{lem:interval-projection-extension} are integrable.
\end{lemma}
\begin{proof}
For $v,c\in[-U,U]$, monotonicity and the normalized envelopes imply
\[
 -UQ_i(v)\le T_i(v)\le vQ_i(v),\qquad
 -UY_j(c)\le T_j(c)\le-cY_j(c).
\]
Thus all transfers remain bounded by $U$ after projection. On every realized
trade in the original market, $v_i>-U$, $c_j<U$, and $v_i>c_j$, so
\[
 |v_i|+|c_j|\le(v_i-c_j)+2U.
\]
At  most $R$   edges trade, and the mechanism's realized gains are bounded by
the profile's first-best gains. Therefore
\[
 \E\!\left[\sum_i |v_i|q_i+\sum_j|c_j|y_j\right]
 \le \FB+2RU<\infty.
\]
Together with bounded transfers, this proves integrability of truthful
utilities and payments.
\end{proof}

\begin{lemma}[Convergence of Projected First Best]
\label{lem:projected-fb-convergence}
For independent Borel priors and downward-closed matching
feasibility,
\[
 \FB_U\uparrow\FB\qquad\text{as }U\to\infty.
\]
\end{lemma}
\begin{proof}
For every edge,
\[
 (v_i^U-c_j^U)_+\uparrow(v_i-c_j)_+.
\]
Downward closedness allows every negative edge to be deleted, so first best is
the maximum feasible sum of these nonnegative edge weights. There are only
finitely many feasible matchings, and each projected matching value increases
to its original value. Thus first best increases profilewise, and monotone
convergence gives the claim.
\end{proof}

\begin{lemma}[Passing Finite Efficiency Guarantees to Borel Priors]
\label{lem:finite-borel-efficiency-transfer}
Fix a compatibility graph, a downward-closed family of feasible matchings,
and $0<\beta\le1$. If every independent finite-prior market on this
feasibility environment satisfies $\SB\ge\beta\FB$, then the same inequality
holds for every independent Borel-prior market with
$\FB<\infty$.
\end{lemma}
\begin{proof}
Fix $U>0$ and a grid on $[-U,U]$ with mesh at most $\delta$. Project all
original types by $\Pi_U$, then round buyers down and sellers up to the grid.
If $\FB^{\rm grid}$ is first best in the resulting finite market, evaluating
a projected first-best matching at the rounded types loses at most
$2\delta$ per edge, so
\[
 \FB^{\rm grid}\ge\FB_U-2R\delta.
\]
By \cref{lem:positive-surplus-finite-optimum}, the finite market has a
BIC/IIR/WBB positive-surplus allocation with GFT at least
$\beta\FB^{\rm grid}$. Extend it to $[-U,U]$ by
\cref{lem:grid-extension-v2}. Because every finite trade has positive rounded
surplus, every trade in this extension has positive actual surplus. Then apply
\cref{lem:interval-projection-extension} to obtain a mechanism for  the original market.
Both extensions preserve expected revenue and weakly increase GFT, while
\cref{lem:interval-projection-integrability} gives the required integrability.
Hence
\[
 \SB\ge\beta\FB_U-2\beta R\delta.
\]
 Let $\delta\downarrow0$   and then $U\to\infty$, using
\cref{lem:projected-fb-convergence}.
\end{proof}

\begin{corollary}[Unbounded Borel Priors]\label{thm:main-v2}
For independent Borel type distributions with $\FB<\infty$
and every downward-closed family of feasible matchings,
\[
 \SB\ge\tfrac12\FB.
\]
The factor is tight.
\end{corollary}
\begin{proof}
Apply \cref{lem:finite-borel-efficiency-transfer} to the finite half guarantee
in \cref{thm:finite-main-v2}. Tightness remains bilateral because finite
bilateral instances are contained in this Borel class.
\end{proof}

\subsection{ { Unbounded  MHR  Priors via Clipping } }\label{app:clipping}
We return to the continuous MHR guarantee and justify its extension to
unbounded supports. Here clipping replaces each type $t$ by
$\min\{t,K\}$; unlike conditioning below a cap, it keeps all probability mass
and moves the tail to $K$. We first extend a bounded mechanism to the
original types, then check integrability and convergence of first best.
The requirement that trade have positive clipped surplus ensures that
sellers above $K$ are never served.

\begin{lemma}
\label{lem:clipping-extension}
Translate the common lower bound to zero and fix $K>0$.  Consider a BIC, IIR,
ex-ante-WBB mechanism defined on the full report domain $[0,K]$ for the
market of clipped types
$v_i^K=\min\{v_i,K\}$ and $c_j^K=\min\{c_j,K\}$.  Suppose its allocation
trades only when $v_i^K>c_j^K$, and use the normalized own-report
interim-envelope transfers on $[0,K]$.
Composing reports with clipping extends it to the original type domains.  The
extension is BIC, IIR, and ex-ante WBB, has the same distribution of payments,
and obtains weakly larger GFT.
\end{lemma}

\begin{proof}
For a buyer with true value $v>K$, compare truthful clipped report $K$ with an
arbitrary clipped report $r\le K$.  Writing $Q,T$ for the clipped interim
service and payment,
\[
 [vQ(K)-T(K)]-[vQ(r)-T(r)]
 =\bigl[KQ(K)-T(K)-KQ(r)+T(r)\bigr]
  +(v-K)\bigl[Q(K)-Q(r)\bigr]\ge0.
\]
The first term is nonnegative by BIC in the clipped market, and the second
is nonnegative by monotonicity. For $v\le K$, the clipped mechanism's BIC
condition applies directly. Truthful utility is also
nonnegative, so buyers satisfy BIC and IIR on the full domain.

A seller with clipped report $K$ is never served by positive-surplus support,
and its normalized envelope payment is zero.  If the true cost is $c>K$, then
for every clipped report $r$,
\[
 -T(r)-cY(r)\le -T(r)-KY(r)\le0,
\]
where the final inequality follows from the clipped mechanism's BIC
condition for type $K$. Thus truthful clipping to $K$ is optimal and gives
utility zero. For $c\le K$, the clipped mechanism's BIC and IIR conditions
apply directly.

Under truthful reports, the clipped type profile has exactly the clipped product
prior, so the payment distribution and expected budget are unchanged.  If an
edge trades, positive clipped surplus rules out $c_j>K$ and gives
\[
 v_i-c_j\ge v_i^K-c_j^K>0.
\]
Hence actual GFT weakly exceeds clipped GFT.
\end{proof}

Clipping preserves incentives, but the unbounded model also requires integrable payments and utilities. Finite expected first best provides the necessary bound.

\begin{lemma}
\label{lem:clipping-integrability}
If the original first-best gains have finite expectation, the normalized
payments and truthful utilities in \cref{lem:clipping-extension} are
integrable.
\end{lemma}

\begin{proof}
Normalized buyer payments lie in $[0,K]$, and normalized seller net payments
lie in $[-K,0]$.  A served seller has original cost below $K$.  On every
realized trade,
\[
 v_i\le (v_i-c_j)+K.
\]
If $R$ is the maximum matching size, the sum of served buyer values is at most
the realized gains of the mechanism plus $RK$.  The mechanism's realized gains
are bounded by realized first best, whose expectation is finite.  Seller costs on served edges are bounded by $RK$, and the sum of absolute
transfers is bounded by the number of agents times $K$.  Therefore utilities and transfers are integrable.
\end{proof}

The final limit concerns only the benchmark. Downward closure makes clipped first-best gains increase to their original value.

\begin{lemma}
\label{lem:capped-fb-convergence}
Let $\FB_K$ be first-best GFT after clipping every value and cost at $K$.  If
types have a common finite lower bound, then after translation to zero,
\[
 \FB_K\uparrow\FB\qquad\text{as }K\to\infty.
\]
\end{lemma}

\begin{proof}
For every edge,
\[
 g_{ij}^K:=\bigl(\min\{v_i,K\}-\min\{c_j,K\}\bigr)_+
 \uparrow (v_i-c_j)_+.
\]
Downward closedness allows every negative edge to be deleted, so profilewise
first best is the maximum feasible sum of these nonnegative edge weights.
It therefore increases pointwise to the original first best.  Monotone
convergence proves the claim.
\end{proof}

\section{Refinements of the Efficiency Bounds}\label{app:refinements}
The local theorem can use more information than a single uniform approximation factor. We first characterize the best uniform factor for a
fixed local support, then allow the target to vary with the seller cap.
We finally strengthen the binary-buyer guarantee using the probability of
the lowest seller cost. These refinements are consequences of the main
reduction, rather than assumptions in either sharp bound.

\subsection{The Best Uniform Local Factor}\label{app:local-factor}
The automatic construction identifies the best factor that can be used
uniformly across all seller caps of a fixed local instance. Use the
buyer support, weakly regular seller, and unnormalized masses of
\cref{sec:cap-monotonicity}.

\begin{corollary}\label{cor:local-coefficient}
If $b_n>s_1$, let
\[
 \kappa_\alpha(\mathbf b,S)=\inf_{r\ge1}\inf_{g\in\Delta_n^\circ}
          \frac{V_\alpha(g,S^{\le r})}{\FB(g,S^{\le r})}.
\]
This is exactly the largest $\beta$ for which a cap-monotone family satisfies
$\mathcal D_\ell^r(z^r)\ge\beta\sum_{k\le r}f_k(b_\ell-s_k)_+$.
The largest coefficient is attained and, on rational inputs, is the value of
a polynomial-size linear program.
\end{corollary}
\begin{proof}[Proof of \cref{cor:local-coefficient}]
The targets are increasing in $r$, so \cref{thm:compatible-targets} proves the
equivalence. Maximize $\beta$ with the displayed linear inequalities,
monotonicity in reports, and cap-monotonicity. There are $O(nm^2)$ allocation variables.
IIR gives $\Lambda\le\GFT\le\FB$, hence $0\le\beta\le1+\alpha$ is enough.
The feasible region is nonempty and compact. If the coefficient fails, minimax
at some cap gives a strictly violating buyer vector; a small positive
perturbation removes any zero masses while preserving the violation.
\end{proof}

\subsection{Bounds That Depend on Seller Caps}\label{ct:app:targets}
A uniform factor takes the worst value over all seller caps. To
  {preserve}  the stronger bounds available at some caps, we instead build
a table of cumulative targets. The compatibility theorem requires these
targets to increase when the cap expands.

Fix the local buyer support and the $\alpha$-weakly regular seller of
\cref{thm:compatible-targets}; $\mathcal D_\ell^r$ denotes the type contribution
in \cref{eq:type-contribution}. The masses in every target are unnormalized.
Assume $b_n>s_1$. For each seller cap define its own worst-case bilateral factor,
without minimizing over the other caps:
\begin{equation}\label{ct:eq:cap-rate}
 \kappa_r=\inf_{g\in\Delta_n^\circ}
 \frac{V_\alpha(g,S^{\le r})}{\FB(g,S^{\le r})},
 \qquad B_\ell^r=\sum_{k\le r}f_k\pos{b_\ell-s_k}.
\end{equation}
Every denominator is positive. Also $0\le\kappa_r\le1+\alpha$, since interim
individual rationality gives $\Lambda\le\GFT$. By the same minimax argument,
$\kappa_r$ is the maximum $\kappa\in[0,1+\alpha]$ satisfying
$\mathcal D_\ell^r(z)\ge\kappa B_\ell^r$ for some $z\in\mathcal X_r$ and all
$\ell$. This is a polynomial-size linear program, and its maximum is attained.

The individual bounds $\kappa_r B_\ell^r$ need not increase with $r$.
Taking a minimum over all larger caps gives the largest increasing
table lying below them, as the next statement formalizes.

\begin{corollary}\label{ct:cor:envelope}
The targets
\begin{equation}\label{ct:eq:envelope}
 C_\ell^r=\min_{q\ge r}\kappa_q B_\ell^q,
 \qquad C_\ell^0=0,
\end{equation}
admit a cap-monotone family. They are the componentwise largest targets that
increase with the cap and satisfy $C_\ell^r\le\kappa_rB_\ell^r$.
In particular, with $\underline\kappa=\min_q\kappa_q$,
\begin{equation}\label{ct:eq:dominate}
 C_\ell^r\ge\underline\kappa B_\ell^r.
\end{equation}
\end{corollary}
\begin{proof}
The minimum in \eqref{ct:eq:envelope} is taken over a shrinking set, so the targets
increase with $r$. They are nonnegative and at most $\kappa_rB_\ell^r$.
The latter target is attainable at every buyer type by the linear program for that
cap, so \cref{thm:compatible-targets} applies. If another increasing
target table $T$ lies below $\kappa_rB_\ell^r$, then for every $q\ge r$,
$T_\ell^r\le T_\ell^q\le\kappa_qB_\ell^q$. Taking the minimum proves
maximality. Finally $B_\ell^q\ge B_\ell^r$ and $\kappa_q\ge\underline\kappa$
prove \eqref{ct:eq:dominate}.
\end{proof}
The maximality statement concerns increasing target tables bounded by
$\kappa_rB_\ell^r$; it does not assert optimality among all matching
mechanisms or all attainable target tables. The construction makes the
cumulative gain targets monotone, which is the property required by the
composition theorem.

\paragraph{The Market-Level Benchmark.}
We next sum the local targets over the edges selected by first best. This
gives a benchmark that can   {preserve}  differences among edges as well as among
seller caps.

Consider any finite independent scalar matching market with downward-closed
feasibility. Fix $\alpha$ and assume every seller is $\alpha$-weakly regular.
Use the canonical first-best matching $M^*$ from \cref{lem:stable-partner-v2}. For each edge
$e=(i,j)$ with possible positive surplus, construct a table $C_{e,\ell}^r$
satisfying \cref{thm:compatible-targets} on its original endpoint priors and buyer
support. In particular, one may use \eqref{ct:eq:envelope}. For residual reports
$\omega=\theta_{-ij}$, let $r_{e,\omega}(\ell)$ be the number of seller types
at which $M^*$ selects $e$ when buyer $i$ reports $b_i^\ell$.
It is a prefix length, increases with $\ell$, and does not depend on the
current seller report. Put
\begin{equation}\label{ct:eq:market-target}
 \mathcal B_\alpha(\mathcal I)=
 \sum_{e=(i,j)}\E_{\omega}\Bigl[\sum_\ell g_{i\ell}C_{e,\ell}^{r_{e,\omega}(\ell)}\Bigr].
\end{equation}
Seller masses are already included in $C$, so no second seller-probability
factor appears in this formula.

\begin{theorem}\label{ct:thm:global}
Under the preceding assumptions,
\[
 V_\alpha(\mathcal I)\ge\mathcal B_\alpha(\mathcal I).
\]
For the envelope targets, if
$\underline\kappa_e=\min_r\kappa_{e,r}$ and
$\FB_e=\E[(v_i-c_j)\1\{e\in M^*\}]$, then
\[
 \mathcal B_\alpha(\mathcal I)\ge
 \sum_e\underline\kappa_e\FB_e.
\]
\end{theorem}
\begin{proof}
For each edge and residual profile, \cref{thm:compatible-targets} supplies
cap-monotone rules with $\mathcal D_\ell^r\ge C_{e,\ell}^r$.
The finite Lagrangian for a fixed seller cap in \cref{eq:fixed-cap-lagrangian} is exactly
$\mathcal D_\ell^r$. Thus \cref{thm:edge-composition} gives
\cref{eq:cap-target-bound}, whose right side is
$\mathcal B_\alpha(\mathcal I)$ by \cref{ct:eq:market-target}.
For the envelope targets, \cref{ct:eq:dominate} bounds each target below by
$\underline\kappa_e B_{e,\ell}^{r_{e,\omega}(\ell)}$.
Averaging over the original endpoint and residual priors turns the latter
surplus into $\underline\kappa_e\FB_e$. Summing proves the second claim.
\end{proof}

\paragraph{Budget and Irregular Sellers.}
The allocation in Theorem~\ref{ct:thm:global} certifies a fixed-multiplier value;
it is not asserted to have nonnegative revenue. If all sellers are regular
for every multiplier, the finite budget dual gives
$\SB(\mathcal I)\ge\inf_{\alpha\ge0}\mathcal B_\alpha(\mathcal I)$.
For arbitrary sellers, apply the full-distribution ironing
transformation at each multiplier to obtain $\widehat{\mathcal I}^{\alpha}$.
It gives the safe statement
\[
 \SB(\mathcal I)\ge
 \inf_{\alpha\ge0}\mathcal B_\alpha(\widehat{\mathcal I}^{\alpha}),
 \qquad V_\alpha(\mathcal I)=V_\alpha(\widehat{\mathcal I}^{\alpha}).
\]
This evaluates the new benchmark in the transformed market. An unweighted
first-best dominance statement does \emph{not} identify the weighted edge
contributions of the original and transformed markets. No efficient global
minimization of this last displayed bound is claimed.

\subsection{Binary Buyers and the Lowest-Cost Probability}\label{app:bottom-mass}
The sharp binary-buyer bound allows the seller's lowest cost to be
arbitrarily rare. A positive lower bound on its probability improves the
guarantee even when the number of seller types is unbounded.

\begin{theorem}\label{thm:bottom-mass}
Fix $p\in(0,1)$. Consider independent finite matching markets with binary
buyers, arbitrary downward-closed feasibility, and
$\Pr[C_j=\min\operatorname{supp}(C_j)]\ge p$ for every seller.
With no bound on seller support size, for every $\alpha>0$ and
$q=\alpha/(1+\alpha)$ the exact worst-case coefficient is
\begin{equation}\label{eq:mass-curve}
 \inf_{\mathcal I}\frac{V_\alpha(\mathcal I)}{\FB(\mathcal I)}
 =\frac{1}{1+q-q^2-q\,p^{(1-q)/q}}.
\end{equation}
The infimum excludes $\FB=0$. The exact second-best coefficient is
\[
 \gamma(p)=\left[\max_{0\le q\le1}
       \{1+q-q^2-q\,p^{(1-q)/q}\}\right]^{-1},
\]
with the continuous endpoint convention at $q=0$.
It satisfies $\gamma(p)>4/5$, tends to $4/5$ as $p\downarrow0$, and tends
to $1$ as $p\uparrow1$. Both infima are approached bilaterally with rational
data satisfying the probability restriction.
\end{theorem}

We prove a stronger finite-support formula first and then let the number
of seller types grow. Only the cutoff calculation changes; the same local
reduction gives the matching guarantee. Fix $p\in(0,1)$ and $\alpha>0$,
and let $q=\alpha/(1+\alpha)$.
For $m\ge2$ define
\begin{equation}\label{eq:mass-P}
 P_{m,p}(q)=q\left(\frac{\alpha}{1+\alpha-p^{1/(m-1)}}\right)^{m-1},
 \qquad D_{m,p}(q)=1+q-q^2-P_{m,p}(q).
\end{equation}
Put $P_{1,p}(q)=q$. At $q=0$ set every $P_{m,p}$ to zero. Let
$\mathfrak M_{2,m}(p)$ be the class $\mathfrak M_{2,m}$ with every seller's
lowest-cost atom of probability at least $p$. The next theorem quantifies
the combined improvement from this probability restriction and the support
bound.

\begin{theorem}\label{thm:mass-finite}
For every $m\ge1$ and $\alpha\ge0$,
\[
 \inf_{\mathcal I\in\mathfrak M_{2,m}(p),\,\FB>0}
       \frac{V_\alpha(\mathcal I)}{\FB(\mathcal I)}
 =\frac{1}{D_{m,p}(q)}.
\]
The same infimum is obtained on bilateral instances and on rational instances.
The exact worst-case second-best coefficient is $1/\max_{0\le q\le1}D_{m,p}(q)$,
where the endpoint values are defined by continuity.
\end{theorem}

The key step strengthens the cutoff recurrence from
\cref{hs:eq:H,hs:eq:K}. Keeping the cumulative probabilities in that
recurrence allows us to use the lower bound on the first atom.

\begin{lemma}\label{lem:mass-cutoff}
Let a seller have $r\le m$ types, with lowest atom $f_1\ge p$.
For the functions $H,K_\alpha$ in \cref{hs:eq:H,hs:eq:K},
  \begin{equation}\label{eq:mass-cutoff}
 H(v)\le (1-P_{m,p}(q))K_\alpha(v).
\end{equation}
No regularity hypothesis is required.
\end{lemma}
\begin{proof}
Let $k_*$ count the costs strictly below $v$. If $k_*=0$, the claim is
immediate. Set $A_k=\sum_{h\le k}f_h(v-s_h)$ for $k\le k_*$. Testing cutoff
$k$ now gives the stronger recurrence
\[
 K_\alpha(v)\ge A_k+\alpha\frac{F_k}{f_k}(A_k-A_{k-1}),\qquad
 A_k\le(1-\theta_k)K_\alpha(v)+\theta_k A_{k-1},\qquad
 \theta_k=\frac{\alpha F_k}{f_k+\alpha F_k}.
\]
Consequently $H(v)\le(1-\prod_{k\le k_*}\theta_k)K_\alpha(v)$.
Here $\theta_1=q$ and, for $k\ge2$,
$\theta_k=\alpha/(1+\alpha-F_{k-1}/F_k)$.

For $m\ge2$, put $x_k=\log(F_k/F_{k-1})$ for $2\le k\le k_*$ and pad the
list to length $m-1$ with zeroes. Its sum is
$\log(F_{k_*}/f_1)\le\log(1/p)$. The function
\[
 h(x)=\log(1+\alpha-e^{-x})\quad(x\ge0)
\]
is increasing and concave: its second derivative is
$-(1+\alpha)e^{-x}/(1+\alpha-e^{-x})^2<0$.
Jensen's inequality therefore yields
\[
 \prod_{k\le k_*}\theta_k
 =\frac{q\alpha^{m-1}}{\prod_{k=2}^{m}(1+\alpha-e^{-x_k})}
 \ge q\left(\frac{\alpha}{1+\alpha-p^{1/(m-1)}}\right)^{m-1}.
\]
Zero padding contributes a factor one to the product of the $\theta_k$.
For $m=1$ that product is just $q$. This proves \cref{eq:mass-cutoff}.
The case $\alpha=0$ is separately $H(v)=K_0(v)$.
\end{proof}

We now substitute this stronger cutoff estimate into the binary-buyer
argument. The probability restriction is preserved under the seller
transformations and conditioning used by the reduction, so the bilateral
bound extends with the same coefficient.

\begin{proof}[Proof of \cref{thm:mass-finite}]
Write $P=P_{m,p}(q)$. In the binary-buyer proof of
\cref{hs:lem:bilateral}, replace each use of $H\le(1-q^m)K_\alpha$ by
\cref{eq:mass-cutoff}. The slope inequality \cref{hs:eq:cutoff-slope} is
unchanged. For a low value $x$, high value $y$, high probability $t$, and
$z=x-qt(y-x)/(1-t)$, the resulting inequalities are
\[
 (1-t)H(x)\le(1-P)(1-t)K_\alpha(z)+q(1-q)tK_\alpha(y),
 \qquad tH(y)\le(1-P)tK_\alpha(y).
\]
Adding gives $\FB\le D_{m,p}(q)V_\alpha$. A deterministic buyer follows
from the cutoff bound alone. The coefficient is positive, since $P\le q<1$.

Full-distribution ironing preserves the ordered atom masses, including the
lowest mass. Conditioning a seller on a prefix of mass $F_r$ changes the
lowest atom to $f_1/F_r\ge f_1\ge p$. Changing the buyer probabilities on the same support preserves the binary
support. Thus all local premises remain in the specified class. Apply
\cref{thm:compatible-targets,eq:cap-target-bound} after full seller ironing,
and return using \cref{eq:regularization-summary}. This proves the matching
lower bound without comparing weighted contributions across transformed markets.

For tightness with $m\ge2$, take the family \cref{eq:family} with
$\eta=p^{1/(m-1)}$ fixed and let $t\downarrow0$. Its lowest seller atom is
exactly $p$, and \cref{eq:family-V,eq:family-F} gives
\[
 \lim_{t\downarrow0}\frac{\FB}{V_\alpha}
 =1+\frac{\alpha}{(1+\alpha)^2}
   -\frac{\alpha}{1+\alpha}
       \left(\frac{\alpha}{1+\alpha-\eta}\right)^{m-1}
 =D_{m,p}(q).
\]
Choose $t$ small enough that the low buyer lies below the second seller cost.
For rational data, approximate $\eta$ from above by rationals below one,
so $\eta^{m-1}\ge p$, and approximate $\alpha$ when necessary. Continuity
of the finite optimum at each fixed positive set of masses preserves the
limit. For $m=1$ the same family has the sole seller cost zero and gives
$D_{1,p}=1-q^2$. At $\alpha=0$ equality is automatic.
Finally $\SB=\inf_\alpha V_\alpha$ and the infima over instances and
multipliers commute. This proves the second-best assertion.
\end{proof}

It remains to remove the support bound. The following limit gives the
uniform guarantee, and the finite tight examples approach that limit.

\begin{proof}[Proof of \cref{thm:bottom-mass}]
For fixed $\alpha>0$, write $a=\log(1/p)>0$. Then
\[
 \lim_{m\to\infty}P_{m,p}(q)
 =q\lim_{h\to\infty}\left(\frac{\alpha}{\alpha+1-e^{-a/h}}\right)^h
 =q e^{-a/\alpha}=q p^{1/\alpha}.
\]
The union over all finite support sizes has coefficient given by this limit.
For completeness, $P_{m,p}$ is nonincreasing in $m$: adding a zero to any
list of logarithmic increments leaves the product unchanged, whereas
redistributing a fixed sum equally maximizes its denominator by the concavity
argument above. Thus every finite-support bound implies the limiting bound,
and the finite tight families approach it as $m\to\infty$.

Define $D_{\infty,p}(q)=1+q-q^2-q p^{(1-q)/q}$ for $0<q<1$.
Its continuous values at zero and one are respectively one and zero.
Taking the two infima gives the displayed formula for $\gamma(p)$.
For any $p>0$, this continuous function is strictly below $5/4$ everywhere:
$1+q-q^2\le5/4$, equality only at $q=1/2$, where the subtracted term is
positive. Hence $\gamma(p)>4/5$. As $p\downarrow0$, its value at $q=1/2$
is $5/4-p/2$, so the maximum tends to $5/4$.
As $p\uparrow1$, the functions converge uniformly on $[0,1]$ to $1-q^2$:
on $[\epsilon,1]$ this follows from the formula, and on $[0,\epsilon]$
the subtracted term is bounded by $q\le\epsilon$. Their maxima therefore
tend to one. Rational finite tight families from \cref{thm:mass-finite}
give the reverse infimum bounds throughout.
\end{proof}

\section{ { Public Values  and Offering Mechanisms } }\label{app:further-scope}
We next apply the reduction beyond the baseline efficiency guarantees.
Public values specific to a match can be incorporated in the local
problems, and   for offering mechanisms  we prove a separate comparison
that translates the Lagrangian bounds into guarantees for those particular
rules.

\subsection{Public Match-Specific Values}\label{app:edgewise-proof}\label{sec:scope}
The scalar private-information model permits different public values for
different matches. Suppose buyer $i$ values item $j$ at $v_i+a_{ij}$,
where $a_{ij}$ is public and only $v_i$ is private; the surplus is then
$v_i-c_j+a_{ij}$. These public offsets preserve the structural properties
used in the reduction.

\begin{theorem}[Binary Endpoints Suffice]\label{thm:edgewise-premia}
For independent finite scalar priors and any downward-closed family of
matchings, suppose every singleton-feasible edge has an endpoint with at most
two types and its other endpoint has at most $m$ types. For arbitrary public
additive values $a_{ij}$,
\[
 V_\alpha\ge D_m\!\left(\frac{\alpha}{1+\alpha}\right)^{-1}\FB,
 \qquad \SB\ge\gamma_m\FB.
\]
With no finite uniform bound on the other endpoint, $\SB\ge4\FB/5$.
The coefficients are tight already when all $a_{ij}=0$ and there is one edge.
\end{theorem}

\begin{proof}
Let $A_i(v)=\E[\sum_j a_{ij}x_{ij}\mid v_i=v]$. Replacing interim buyer
payment $T_i(v)$ by $T_i(v)-A_i(v)$ gives the scalar incentive inequalities.
Thus the normalized maximum revenue is
\[
 \Lambda_a(x)=\E\Bigl[\sum_{ij}(\phi_i(v_i)-\tau_j(c_j)+a_{ij})x_{ij}\Bigr].
\]
The Lagrangian edge weight is $a_i^\alpha(v_i)-d_j^\alpha(c_j)+(1+\alpha)a_{ij}$.
Public constants only change the intercepts of served matchings. Stable
partners, cap-monotonicity, and the ironing equalities continue to hold. Conditional
first best is still $\max\{A,B-c\}$ in one seller's cost, so the same full
seller replacement preserves $V_\alpha$ and weakly increases $\FB$.

At edge $(i,j)$ use the shifted buyer support $b_i^\ell+a_{ij}$. Its gaps and
number of types are unchanged. If the binary endpoint is the seller, swap and
reflect the two sides in the bilateral inequality: for a common constant $U$,
$\phi_{U-C}(U-c)=U-\tau_C(c)$ and
$\tau_{U-B}(U-b)=U-\phi_B(b)$. The resulting support bounds are exactly those
of \cref{hs:lem:bilateral}. Changing the buyer probabilities on the same support and conditioning the
seller on a lower prefix preserve them,
as does the full seller replacement. Therefore every local premise holds at
the common coefficient. The rent correction \cref{eq:rent-correction} is
unchanged. Edge-by-edge composition and duality prove the result. Any finite instance has a
finite maximum support size, so letting this size vary gives the $4/5$ statement.
\end{proof}

\subsection{Generalized Offering Mechanisms}\label{app:offering-consequences}
The reduction preserves a value bound, not a chosen bilateral mechanism.
For generalized offering mechanisms there is, however, a useful additional
comparison at $\alpha=1$, relating the optimized Lagrangian to the gains
of two specified rules.

Use independent finite priors and downward-closed matching feasibility.
Let $\bar\phi_i$ and $\bar\tau_j$ be the separately ironed virtual values and
virtual costs. The generalized seller-offering mechanism (GSOM) maximizes
$\sum_{(i,j)\in M}(\bar\phi_i(v_i)-c_j)$; the generalized buyer-offering
mechanism (GBOM) maximizes
$\sum_{(i,j)\in M}(v_i-\bar\tau_j(c_j))$. In both rules minimize cardinality
among maximizers, then use a fixed order, and use the finite-support threshold
payments. These are the mechanisms of \cite{BCWZ17}; see also
\cite{BRTW26}. Write $G_{\rm offer}$ for the expected GFT when choosing
between them with equal probability, independently of reports.

\begin{lemma}[Comparison with Offering Mechanisms]\label{lem:offering-comparison}
Both rules are DSIC, ex-post IR, and ex-ante WBB, and
\begin{equation}\label{eq:offering-comparison}
 2G_{\rm offer}\ge V_1.
\end{equation}
\end{lemma}
\begin{proof}
The implementation properties are established in
\cite[Lemma~10]{BCWZ17}. We recall the argument for our finite-support
conventions. The intercept argument of \cref{lem:exact-ironed-optimizer}
gives monotone service and a fixed partner throughout each winning region.
Threshold payments therefore give DSIC and ex-post IR. In GSOM a served
seller's threshold receipt is at most her partner's $\bar\phi_i(v_i)$.
Expected buyer payments equal
$\E[\sum_i\phi_i(v_i)q_i]=\E[\sum_i\bar\phi_i(v_i)q_i]$, because service is
constant on each buyer ironing block. Expected receipts cannot exceed these
payments. For GBOM, each served buyer's threshold payment is at least her
partner's $\bar\tau_j(c_j)$; expected seller receipts equal
$\E[\sum_j\bar\tau_j(c_j)y_j]$. This proves ex-ante WBB for both rules.

For any $x\in\mathcal X$, split its budget Lagrangian into
\[
 \GFT(x)+\Lambda(x)
 =\E\Bigl[\sum_{(i,j)}(\phi_i(v_i)-c_j)x_{ij}\Bigr]
  +\E\Bigl[\sum_{(i,j)}(v_i-\tau_j(c_j))x_{ij}\Bigr].
\]
The two ironing inequalities bound these terms by, respectively,
\[
 \E\Bigl[\max_{M\in\mathcal F}\sum_{(i,j)\in M}
       (\bar\phi_i(v_i)-c_j)\Bigr],\qquad
 \E\Bigl[\max_{M\in\mathcal F}\sum_{(i,j)\in M}
       (v_i-\bar\tau_j(c_j))\Bigr].
\]
These are the objectives optimized by GSOM and GBOM. Since
$\bar\phi_i(v_i)\le v_i$ and $\bar\tau_j(c_j)\ge c_j$, their actual gains
are at least these respective objectives. Maximizing the left side over
$x$ proves \cref{eq:offering-comparison}.
\end{proof}

Consequently, any factor $\beta$ satisfying the premise of
\cref{thm:bilateral-to-matching} at $\alpha=1$ implies
$G_{\rm offer}\ge(\beta/2)\FB$. This is a consequence of the reduction
\emph{and} \cref{lem:offering-comparison}, not of the reduction alone.
Applying this comparison to the unrestricted fixed-multiplier bound and
to our binary-buyer curve gives the following guarantees.

\begin{corollary}[Offering Guarantees]\label{cor:offering-guarantees}
Under independent finite priors and downward-closed matching feasibility,
$G_{\rm offer}\ge\FB/\pi$. If all buyers are binary and all sellers have at
most $m$ types, then
\begin{equation}\label{eq:binary-offering-bound}
 G_{\rm offer}\ge\frac{2}{5-2^{2-m}}\FB.
\end{equation}
In particular, binary types on both sides give $\FB/2$, and binary buyers
with arbitrary finite seller supports give $2\FB/5$. The statements also
hold with the two sides exchanged. No tightness claim is made for these
mechanism guarantees.
\end{corollary}
\begin{proof}
For the unrestricted statement, the finite-prior inequality
$V_1\ge(2/\pi)\FB$ follows from the fixed-multiplier bound in
\cite[Lemma~4.9]{BLWZ26}. Its one-edge case also supplies the premise of
\cref{thm:bilateral-to-matching} at $\alpha=1$. The unrestricted input is therefore a prior result; the comparison lemma
explains its consequence for these mechanisms.
For the support restriction, substitute $q=1/2$ in \cref{thm:curve} to get
$V_1\ge(5/4-2^{-m})^{-1}\FB$, then apply
\cref{lem:offering-comparison}. Letting the finite support size grow gives
$2/5$; reflection exchanges the sides.
\end{proof}

Jo's bilateral analysis~\cite{Jo26} obtains the same $2/\pi$ Lagrangian
factor at $\alpha=1$ and converts it to the $1/\pi$ random-offerer guarantee
by comparison with proposer profits. His fixed-multiplier theorem is stated
for smooth priors, followed by a separate approximation argument for the
random-offerer guarantee on bounded Borel priors. The finite-prior version used here is supplied by \cite{BLWZ26}. A bilateral upper bound for random
offerer bounds that specified mechanism only; it gives no upper bound on
$\SB$ and is not an output of our reduction.

\end{document}